\documentclass[12pt]{article}
\usepackage{amssymb}
\usepackage[nohead]{geometry}
\usepackage{setspace}
\usepackage[colorlinks=false]{hyperref}
\usepackage{amsthm} 
\usepackage{natbib}
\usepackage{amsmath}
\usepackage[dvipsnames]{xcolor}
\usepackage{amsthm}
\usepackage{sgame}
\usepackage{hyperref}
\usepackage{cleveref}
\usepackage{verbatim}
\usepackage{subcaption}
\usepackage[T1]{fontenc}
\usepackage{enumitem}
\usepackage{subfiles}
\usepackage{lmodern}

  \theoremstyle{plain}
  
 \theoremstyle{definition}
  \newtheorem{example}{\protect\examplename}
  \theoremstyle{plain}
  \newtheorem{prop}{\protect\propositionname}
  \theoremstyle{plain}
  
  \theoremstyle{plain}
  \newtheorem{assumption}{\protect\assumptionname}
  \theoremstyle{plain}
  \newtheorem{lem}{\protect\lemmaname}
\theoremstyle{plain}
\newtheorem{thm}{\protect\theoremname}

\theoremstyle{definition}
  \newtheorem{defn}{\protect\definitionname}

  \theoremstyle{remark}

  \theoremstyle{definition}

\defcitealias{HSTV_2011_ECMA}{HSTV}

\makeatother

  \providecommand{\assumptionname}{Assumption}
  \providecommand{\claimname}{Claim}
  \providecommand{\conjecturename}{Conjecture}
  \providecommand{\definitionname}{Definition}
  \providecommand{\examplename}{Example}
  \providecommand{\lemmaname}{Lemma}
  \providecommand{\propositionname}{Proposition}
\providecommand{\corollaryname}{Corollary}
\providecommand{\theoremname}{Theorem}

\begin{document}

\title{Asymptotic value of monitoring structures\\in stochastic games\thanks{We are grateful to the co-editor and three anonymous referees for their invaluable suggestions. We wish to thank Satoru Takahashi for his helpful comments, and also thank seminar participants at the 6th World Congress of the Game Theory Society, the 2020 Econometric Society World Congress, and the HKBU-NTU-Osaka-Kyoto Theory Seminar for their feedback. All remaining errors are our own.}
\author{Daehyun Kim\thanks{Department of Economics, UCLA. E-mail: \href{mailto:dkim85@outlook.com}{\nolinkurl{dkim85@outlook.com}}}\and Ichiro Obara\thanks{Department of Economics, UCLA. E-mail: \href{mailto:iobara@econ.ucla.edu}{\nolinkurl{iobara@econ.ucla.edu}}}}}

\date{October 20, 2025}

\maketitle

\singlespacing

\begin{abstract}
This paper studies how improved monitoring affects the limit equilibrium payoff set for stochastic games with imperfect public monitoring. We introduce a simple generalization of Blackwell garbling called \emph{weighted garbling} in order to compare different monitoring structures for this class of games. Our main result is the monotonicity of the limit perfect public equilibrium (PPE) payoff set with respect to this information order. We show that the limit PPE payoff set expands when the monitoring structure gets more informative with respect to the weighted garbling order. We also show that a similar monotonicity holds for strongly symmetric equilibrium for symmetric stochastic games. Finally, we show that our weighted garbling order is useful to compare the limit PPE payoff set for different state transition laws and monitoring structures when the limit feasible payoff set is the same.
\end{abstract}
\strut

\textbf{Keywords:} Comparison of experiments, garbling, imperfect monitoring, monitoring structure, perfect public equilibrium, stochastic games 

\strut

\textbf{JEL Classification Numbers:} C72, C73 

\pagebreak

\pagebreak

\onehalfspacing

\section{Introduction}
\label{sec:1}

This paper examines the informativeness of monitoring structures in stochastic games with imperfect public monitoring, where the players' actions are monitored imperfectly through noisy public signals. Such models have broad economic relevance. For example, many standard repeated moral hazard scenarios (e.g., tacit collusion) can be framed in such stochastic games when the key underlying state, like a demand shock or an interest rate, evolves over time. In these settings, the quality of monitoring structure plays a critical role; it shapes agents' incentives and, consequently, determines the range of possible outcomes that could arise in equilibrium.  

We study how a concept of garbling can be used to compare the \emph{limit} equilibrium payoffs across different monitoring structures. Specifically, we introduce a novel form of garbling, termed \emph{weighted garbling}, and demonstrate that it captures the informativeness of monitoring structures when players are very patient.

Blackwell ordering is one of the most important information orders in Economics and Statistics. It is equivalent to the existence of garbling (noise) that transforms a more informative information structure into a less informative one \cite[]{Blackwell_1951, Blackwell_1953_AMS}. Weighted garbling represents a more general class of such transformations, encompassing the standard Blackwell garbling as a special case.

In repeated games with imperfect public monitoring, it is well established that the set of perfect public equilibrium (PPE) payoffs \cite[]{FLM_1994_ECMA} expands for any fixed discount factor as the monitoring structure becomes more informative under the Blackwell order \cite[]{Kandori_1992_RES}. However, this monotonicity does not extend to stochastic games \cite[]{Kim_2019_IJGT}.

Our first main result establishes that the set of PPE payoffs that can be sustained from any initial state expands weakly in the limit---as discounting vanishes---when the monitoring structure becomes more informative under the weighted garbling order, assuming a certain kind of full-dimensionality of payoffs. In particular, if the Markov chain over states induced by any pure Markov strategy profile is irreducible, then a state-independent limit PPE payoff set exists and it expands with increased informativeness in the weighted garbling order.

Since the standard garbling, which we refer to as \emph{joint garbling} in the context of stochastic games, is a special case of weighted garbling, a corollary of our result is that the limit PPE payoff set expands when the monitoring structure becomes more Blackwell informative. This implies that the standard Blackwell ordering remains a meaningful measure of informativeness for stochastic games in the limit, even if it fails to guarantee expansion for fixed discount factors. Furthermore, by focusing on the limit, our result can establish an expansion of the equilibrium payoff set across a broader class of monitoring structure pairs under weighted garbling than under Blackwell (or joint) garbling.

Notably, our result yields a novel insight even for repeated games. It demonstrates that in the limit case, the equilibrium payoff set expands as the monitoring structure becomes more informative under the weighted garbling order, even if it does not become more informative in the sense of Blackwell.

To facilitate the discussion of weighted garbling, let us provide its precise definition. \emph{Monitoring structure} $(Y, f)$ consists of a finite set $Y$ of public signals and a conditional distribution $f( \cdot |t, s, a) \in \Delta(Y)$, which depends on the current state $s$, the next state $t$, and the current action profile $a$. Combined with the state transition law $q( \cdot |s,a) \in \Delta(S)$, which is a primitive of the model that remains fixed throughout most of the paper (except for \Cref{sec:6}), this monitoring structure generates a joint distribution $p(\cdot, \cdot|s,a) \in \Delta(S \times Y)$ over the next state-signal pair $(t, y)$ given $(s,a)$. Since both the next state and public signal convey information about the hidden action $a$, it is natural to evaluate the informativeness of the joint distribution with respect to actions at each state $s \in S$. A monitoring structure $(Y, f)$ is said to be a weighted garbling of $(Y', f')$ if, for each state $s$, there exist a nonnegative \emph{weight} $\gamma_{s}^{t', y'} \geq 0$ and a \emph{garbling} $\phi_s(\cdot, \cdot|t', y') \in \Delta(S \times Y)$ for each $(t',y')$ such that $p(t,y|s,a) = \sum_{(t', y') \in S \times Y'}\gamma_{s}^{t', y'} \phi_s(t, y| t', y') p'(t',y'|s,a)$ for every next state-signal pair $(t,y)$ and for every action profile $a$. As a special case, if the weights $\gamma_{s}^{t', y'}$ are uniformly equal to $1$, the condition simplifies to: $p(t,y|s,a) = \sum_{(t', y') \in S \times Y'} \phi_s(t, y| t', y') p'(t',y'|s,a)$, which corresponds to the standard Blackwell garbling, referred to as joint garbling.

The simplest example of weighted garbling can be found in repeated games with imperfect public monitoring. Consider a monitoring structure that produces some public signal $\tilde y$. Now, imagine an alternative monitoring structure where no signal is observed with probability $1-\epsilon$ and a strictly more Blackwell informative public signal $\hat y$ is observed with probability $\epsilon \in (0,1)$. Provided that $\epsilon$ is sufficiently small, neither monitoring structure is more informative than the other in the sense of Blackwell. However, conditional on observing $\hat y$, the second monitoring structure is clearly more informative. In this case, the first monitoring structure is a weighted garbling of the second, but not vice versa. In this simple example, it is intuitively clear why the limit PPE payoff set is larger under the second monitoring structure: if players are very patient, they can disregard the no-signal events and wait for the more informative signal to permit more efficient punishment. Our main result formally establishes such monotonicity of limit PPE payoffs with respect to the weighted garbling order for stochastic games with imperfect public monitoring.

Our monotonicity result builds on two key ideas for using information more efficiently when players are very patient. The first idea is to treat both the next-period state and the public signal jointly as a monitoring device. Since the next state conveys information about current actions, it is natural to evaluate the informativeness of the joint distribution over the next period state and public signal. However, the next state also determines the physical state of the game, so it cannot be treated purely as a signal. For example, even if the next state reveals a defection in the current period, it may not allow for a sufficiently severe punishment to deter such behavior.
To address such state-dependent constraints, we focus on the set of payoffs that can be sustained at every state. It is known that each such payoff profile can be supported by constructing a state-independent ``block self-generating set'' around it (with respect to joint distribution) when $\delta$ is sufficiently high \cite[henceforth HSTV]{HSTV_2011_ECMA}. When the joint distribution becomes more informative in the sense of standard garbling (i.e., joint garbling), one can mimic this construction by treating the next state as a purely informative signal and modifying the continuation payoffs in the original construction. Although these new continuation payoffs differ from the original ones in a complex way, when $\delta$ is sufficiently high, they behave similarly and stay in the same state-independent self-generating set. This explains why the set of PPE payoffs that can be sustained at every state expands weakly in the limit as $\delta \rightarrow 1$ when the monitoring structure becomes more informative under the joint garbling. Furthermore, if some irreducibility condition is satisfied, then this limit set of state-independent PPE payoffs coincides with the limit PPE payoff set at every state. The second idea centers on conditional informativeness, as illustrated in the repeated game example above. When players receive a strictly more informative signal but with a small probability, it may be impossible to sustain the same level of cooperation at a fixed discount factor. However, with greater patience, the punishment becomes effective enough to support cooperation. Thus, more patient players can exploit such conditionally more informative signals to enhance efficiency.

Our second main result focuses on strongly symmetric equilibrium (SSE) for symmetric stochastic games. SSE is an important class of PPE in symmetric repeated or stochastic games, where every player chooses the same action after every public history. As a natural and tractable refinement of PPE, SSE is often the focus of applied theorists and economists. 
We establish that a version of weighted garbling, adapted specifically for symmetric environments, yields monotonicity in the limit SSE payoff set. This result is particularly valuable because no folk theorem applies in SSE. Recall that the first-best outcome is usually impossible to sustain in SSE due to the necessity of costly group punishments.\footnote{Such inefficiency arises even for general PPE when individual deviations cannot be distinguished in symmetric settings \cite[]{RMM_1986_ECMA}.} 
Our result demonstrates that a more informative monitoring structure under the weighted garbling order can generate a larger limit SSE payoff set by mitigating such incentive cost.

We apply our result for SSE to stochastic partnership games, where the team outcome depends solely on the total effort of its members. We show that when the agents are very patient, the best sustainable team outcome under SSE strictly improves as the monitoring structure becomes more informative under a strict version of the weighted garbling order.

Throughout this paper, we focus on the case of arbitrarily patient players rather than the case of fixed discount factors, for the same reason as we focus on the limit PPE payoff set in folk theorems. In many settings, high levels of patience allow us to abstract from temporary forces and transient dynamics, hence obtain sharper results regarding the comparison of various monitoring structures. As such, our findings offer valuable guidance for those who use stochastic games in the settings that involve very patient agents.

\vspace{5mm}

\textbf{Related Literature}

\vspace{3mm}

\cite{Kandori_1992_RES} and \cite{Kim_2019_IJGT} are the most closely related to our work. Both examine how monitoring structures affect the PPE payoff set under a fixed discount factor. \cite{Kandori_1992_RES} analyzes repeated games with imperfect public monitoring and shows that the PPE payoff set expands as the monitoring structure becomes more informative in the sense of Blackwell, leveraging the recursive characterization developed by \cite{APS_1990_ECMA}.\footnote{The original result of \cite{Kandori_1992_RES} assumes continuous signals and uses pure-strategy sequential equilibrium as the solution concept. It is well known that, in such an environment, any pure-strategy sequential equilibrium has a payoff-equivalent pure-strategy PPE. In contrast, we assume finite signals and allow for mixed-strategy PPE.} In the context of stochastic games, \cite{Kim_2019_IJGT} provides an extension of Kandori's condition that guarantees the expansion of the PPE payoff set for a fixed discount factor. This condition relies on the notion of \emph{ex-post garbling}, which is stronger than joint garbling and requires that the public signal be more informative for every state transition pair $(s,t)$. Importantly, \cite{Kim_2019_IJGT} also presents a counterexample showing that a more informative monitoring structure under the standard (joint) garbling order does \emph{not} guarantee expansion of the equilibrium payoff set for a \emph{fixed} discount factor. In contrast, our main result shows that a more informative monitoring structure under joint garbling \emph{does} expand the \emph{limit} equilibrium payoff set, since joint garbling is a special case of weighted garbling.

This study also connects to the folk theorem literature for repeated and stochastic games, as it focuses on the limit PPE payoff set. \cite{FY_2011_JET} and HSTV establish a variety of folk theorems for stochastic games with imperfect public monitoring.\footnote{An earlier contribution by \cite{Dutta_1995_JET} offers a folk theorem and foundational results for stochastic games with perfect monitoring.} We build on the characterization of limit PPE payoff set by HSTV, which extends the approach of \cite{FL_1994_JET} for repeated games to stochastic games. 

Our main results are particularly valuable in environments where the folk theorem does not apply. For example, in repeated games, the folk theorem holds when the monitoring structure is sufficiently rich and satisfies identifiability conditions such as the pairwise full rank condition \cite[]{FLM_1994_ECMA}. In such cases, all feasible and individually rational payoffs can be sustained without any incentive cost, making it impossible to expand the limit PPE payoff set even with strictly more informative monitoring structures.

Several studies highlight that improved information can sometimes be detrimental in specific classes of stochastic games. \cite{Kloosterman_2015_JET} examines a setting where state transitions are independent of players' actions, actions are observable, and players observe a public signal about the next state before choosing their actions. He shows that, for a given discount factor, the PPE payoff set may shrink as the public signal becomes more informative about the next state. This model can be interpreted as a special case of our model with observable actions, by treating the pair of current state and public signal as the effective current state. With observable actions, our weighted garbling order does not rank any two information structures, as it is designed to measure the informativeness of monitoring structures with respect to hidden actions. \cite{Sugaya_Wolitzky_2018_JPE} study a dynamic price-setting oligopoly with \emph{private} monitoring, where the demand state evolves exogenously. They show that increased transparency can undermine collusion, meaning that more informative private signals may reduce cooperation. In contrast, our results demonstrate that a more informative \emph{public} signal always expands the limit equilibrium payoff set, reconfirming the role of informative public information in sustaining cooperation among patient players.

A related strand of literature explores the trade-off between the quality of monitoring and the patience of players \cite[]{AMP_1991_ECMA, FL_2007_RED, FL_2009_QJE, Sannikov_Skrzypacz_2007_AER, Sannikov_Skrzypacz_2010_ECMA, Sugaya_Wolitzky_2022_WP}. For example, such a trade-off naturally arises in repeated games with frequent actions and a continuous flow of information. As the period length shortens, players become more patient as they place greater weights on future payoffs, while the quality of monitoring may simultaneously decline. In contrast, our approach focuses directly on the limit case as the discount factor approaches $1$, a standard approach in folk theorem analysis, to compare monitoring structures.

The remainder of the paper is organized as follows. \Cref{sec:2} introduces our model of discounted stochastic games with imperfect public monitoring. \Cref{sec:3} defines weighted garbling and related notions, illustrates them with examples, and provides a sufficient condition for weighted garbling in a special class of stochastic games. In \Cref{sec:4}, we present our main results for general stochastic games, with the central finding being the monotonicity of the limit PPE payoff set with respect to the weighted garbling order. In \Cref{sec:5}, we establish similar results for the limit SSE payoff set, using a version of weighted garbling tailored to symmetric environments, and apply them to stochastic partnership games to demonstrate the relevance of the concept in applications. \Cref{sec:6} explores how our results can be used to compare two monitoring structures with different transition laws that yield the same limit feasible payoff set. Finally, \Cref{sec:7} concludes.

\section{Model}
\label{sec:2}

We study discounted stochastic games with imperfect public monitoring. Let $S$ be the finite set of states, and let $I\equiv \{1,2,\dots, N\}$ be the finite set of players. For each player $i \in I$, let $A_i$ be the finite set of actions available to player $i$, and let $A \equiv \prod_{i \in I} A_i$. Player $i$'s payoff function is given by $u_i: A \times S \to \mathbb{R}$. The tuple $G \equiv(I, (A_i)_i, (u_i)_i)$ defines the stage game at each state.

Each period $k \in \mathbb{N}$ starts with a state $s \in S$. Each player $i$ takes an action $a_i \in A_i$, then the new state $t \in S$ for the next period is drawn randomly according to the transition law $q (\cdot |s,a) \in \Delta (S)$, where $a = (a_i)_i$. When the transition law does not depend on the action, we say that the transition is \emph{action-independent} and use the notation $q(\cdot|s)$.
Players share a common discount factor $\delta \in \left[0,1\right)$ to discount future payoffs. The triple $(G,q, \delta)$ defines the underlying dynamic strategic environment, except for the monitoring structure, which we introduce next.

States are perfectly observed by all the players at the beginning of each period, but actions are not directly observable. Instead, players observe a public signal about actions. Thus, we augment $(G,q, \delta)$ with a \emph{monitoring structure} $\Pi=(Y,f)$,  which defines the monitoring environment of the game. In each period, players observe a public signal $y \in Y$ from a finite set $Y$. The distribution of public signal depends on the next state $t \in S$, the current state $s \in S$, the current action profile $a \in A$, and is denoted by $f ( \cdot |t, s, a) \in \Delta (Y)$. We can allow perfect monitoring. For example, $\Pi$ such that $Y = A$ and $f(a|t, s ,a) =1$ for any $s,t \in S$ and $a \in A$ is a perfect monitoring structure. The pair $\Pi$ and $q$ induces a joint distribution $p (\cdot, \cdot |s,a) \in \Delta (S \times Y)$ conditional on any $(s,a) \in S \times A$, defined by $p(t,y|s,a) := f(y|t,s,a)q(t|s,a)$ for any $(t,y) \in S \times Y$. 

We assume that players observe only the public signal, the state, and their own action (privately) in each period. This does not necessarily mean that realized payoffs are unobservable. In many IO applications of repeated games, there is a realized payoff $g_i(a_i ,y)$, which is observable to player $i$ and depends only on player $i$'s action and the public signal $y$. Then the expected payoff $u_i(a)$ becomes a function of the action profile. In this case, player $i$'s realized payoff does not provide any additional information about $a_{-i}$ beyond $y$. Similarly, we may introduce a realized payoff $g_i(a_i, y, t, s)$ and derive the above $u_i$ as $u_i (a,s) = \sum_{ (t,y) \in S \times Y} g_i (a_i,y,t,s)p (t,y|s,a)$ in certain contexts.\footnote{Note that, if realized payoff $g_i$ does not depend on $y$ (i.e., $g_i (a_i,s,t)$), we can change the monitoring structure without affecting $u_i$. This contrasts with repeated games, where monitoring structure and ex-ante expected payoff become dependent under this type of interpretation. In \cite{Kandori_1992_RES}, realized payoffs need to be adjusted to maintain the same expected payoffs across different monitoring structures. Such an adjustment is not necessarily needed in our setting due to the additional flexibility of state-dependent payoffs.}

A \emph{discounted stochastic game with imperfect public monitoring} is defined as a quadruple $\Gamma = (G, q, \delta, \Pi)$. Henceforth, we fix $(G,q)$ and vary $\left(\delta, \Pi\right)$ to address our research questions, except in \Cref{sec:6}.

A private history of player $i$ at period $k$ is a sequence of realized states, chosen actions, and realized public signals up to the beginning of period $k$, i.e., $h_i^k = (s^1, a_i^1, y^1,\dots, s^{k-1},a_i^{k-1},y^{k-1},s^k)$. A public history at period $k$ is $h^k =(s^1, y^1,\dots, s^{k-1},y^{k-1},s^k)$. Let $H_i^k$ be the set of all private histories of player $i$ at period $k$ and $H^k$ be the set of all public histories at period $k$. Furthermore, let $H_i \equiv \bigcup_{k=1}^\infty H_i^k$ and $H \equiv \bigcup_{k=1}^\infty H^k$.

A (behavioral) strategy $\sigma_i$ for player $i$ is defined as a mapping from $H_i$ to $\Delta (A_i)$. Let $\Sigma_i$ denote the set of strategies for player $i$. A strategy is \emph{public} if it depends only on public histories. Given a strategy profile $(\sigma_i, \sigma_{-i})$, player $i$'s average discounted payoff is 
$$U_i (\sigma_i, \sigma_{-i}; s) : = \left(1-\delta\right) \mathbb{E}^{(\sigma_i, \sigma_{-i})} \left[ \sum_{ k =1}^\infty \delta^{k-1} u_i (a^k, s^k) | s^1 = s \right],$$
where the expectation is evaluated using the probability measure $\mathbf{P}$ over $(S \times A)^\infty$ induced by $(\sigma_i)_{i \in I}$ and initial state $s$.

A profile of public strategies $(\sigma_i)_{i \in I}$ is a \emph{perfect public equilibrium (PPE)} \cite[]{FLM_1994_ECMA} if at every public history, the continuation public strategy profile is a Nash equilibrium for the continuation game. That is, for each $i$ and $h \in H$,
$$U_i (\sigma_i |_{h} , \sigma_{-i} |_h ; s(h) ) \geq U_i ( \sigma_i', \sigma_{-i} |_{h} ; s(h)), \quad \forall \sigma_i' \in \Sigma_i$$
where $\sigma_i |_{h}$ is the continuation strategy of $\sigma_i$ at public history $h \in H$ and $s(h)$ is the most recent state at  $h$.\footnote{In general, PPE is a strict subset of sequential equilibrium for games with imperfect public monitoring. For examples of sequential equilibria where players use non-public strategies (i.e., private strategies), see \cite{Kandori_Obara_2006_ECMA} and \cite{Mailath_Matthews_Sekiguchi_2002_BE}.}

\section{Comparison of Monitoring Structures}
\label{sec:3}
\subsection{Weighted Garbling}
\label{subsec:3.1}

In this section, we introduce \emph{weighted garbling}, a new notion of garbling that is the focus of this paper. We begin by discussing a few special cases of weighted garbling introduced in \cite{Kim_2019_IJGT} to clarify the basic ideas behind it.

We first observe that $\Pi' = (Y',f')$ is clearly more informative about actions than $\Pi = (Y,f)$ if, for each $s,t \in S$, there exists $\phi_{s,t}:Y' \rightarrow \Delta(Y)$ such that 
$$ f(y|t,s,a) = \sum_{y' \in Y'}\phi_{s,t}(y|y') f'(y'|t,s,a),\quad \forall y \in Y$$ for every $a \in A$. This means that the public signal for $\Pi$ is a garbling of the public signal for $\Pi^\prime$ conditional on every state transition $(s,t) \in S \times S$. \cite{Kim_2019_IJGT} introduces this notion of garbling, called \emph{ex-post garbling}, and shows that the set of PPE payoffs (weakly) expands for any fixed discount factor when a monitoring structure improves in this sense. This is an extension of the well-known result for repeated games \cite[]{Kandori_1992_RES} to stochastic games, as ex-post garbling reduces to the standard garbling for repeated games.

Ex-post garbling is a very strong condition in stochastic games. Instead, we may prefer to use the standard garbling without conditioning each state transition to rank monitoring structures. \emph{joint garbling} \cite[]{Kim_2019_IJGT} is a weaker notion of garbling and only requires that the pair of next period state $t$ and public signal $y$ for $\Pi$ is a garbling of $t$ and $y^\prime$ for $\Pi^\prime$. Formally, $\Pi$ is a \emph{joint garbling} of $\Pi^\prime$ if, for each $s \in S$, there exists $\phi_s : S \times Y' \to \Delta (S \times Y)$ such that
\[
p(t,y|s,a)= \sum_{(t^\prime,y^\prime) \in S \times Y^\prime} \phi_s(t,y|t^\prime,y^\prime)p'(t^\prime,y^\prime|s,a), \quad \forall (t,y)\in S \times Y
\] 
for every $a \in A$.\footnote{This does not mean that states need to be literally reshuffled to obtain $p$ from $p'$. It just means that $(t, y)$ can be interpreted as $(t^\prime, y^\prime)$ plus noise (like the standard garbling). See \autoref{exmp:WG}.} 

Joint garbling is a natural concept, as the next period state and public signal are jointly informative about the current actions. After all, it is simply the standard Blackwell garbling of a random vector. However, \cite{Kim_2019_IJGT} shows that the PPE payoff set may shrink given a fixed discount factor when the monitoring becomes more informative in the joint garbling order.

It is useful to understand why the monotonicity with respect to the joint garbling order fails for stochastic games. For repeated games (with a public randomization device), the monotonicity with respect to the standard garbling order holds for the following reason. Suppose that a payoff profile $v \in \mathbb{R}^N$ can be supported by an action profile $a$ and continuation payoffs $w:Y \rightarrow E^\delta(\Pi)$, where $E^\delta(\Pi)$ is the PPE payoff set given $\delta$ and $\Pi$. Suppose that $\Pi$ is a garbling of $\Pi^\prime$; then there exists $\phi: Y' \rightarrow \Delta(Y)$ such that $p(y|a) = \sum_{y' \in Y'}\phi(y|y') p(y'|a)$. Then, $v$ can be supported by the same action profile and continuation payoffs $w': Y' \rightarrow \mathbb{R}^N$ defined by $w'(y') = \sum_{y \in Y}\phi(y|y') w(y)$ for $\Pi^\prime$. Since the equilibrium payoff set is convex (with a public randomization device), this means that $E^\delta(\Pi)$ is self-generating for $\Pi^\prime$ as well. Hence, the equilibrium payoff set $E^\delta(\Pi^\prime)$ must contain $E^\delta(\Pi)$. This argument does not fully extend to stochastic games. Suppose that $v$ can be supported at state $s$ by an action profile $a$ and continuation payoffs $w: S \times Y \rightarrow \prod_{t \in S}E^\delta(t; \Pi)$, where $E^\delta(t; \Pi)$ is the PPE payoff set at state $t$ given $\delta$ for $\Pi$. Also, suppose that $\Pi$ is a joint garbling of $\Pi^\prime$. We can define continuation payoffs $w': S \times Y' \rightarrow \mathbb{R}^N$ in a similar way by $w'(t',y') = \sum_{(t,y) \in S \times Y}\phi_s(t,y|t',y') w(t,y)$. Although all the incentive constraints are preserved as before, $w'(t',y')$ may not belong to $E^\delta(t'; \Pi)$ because the equilibrium payoff sets differ across states. Hence $\prod_{t \in S}E^\delta(t; \Pi)$ in stochastic games is not self-generating for $\Pi'$ unlike $E^\delta(\Pi)$ in repeated games. Intuitively, the issue is that the next-period state is not only an informative signal but also affects the feasibility of continuation payoffs in equilibrium.

In this paper, we introduce the following simple yet even weaker notion of garbling, called \emph{weighted garbling}, and show that it captures the informativeness of monitoring structures for stochastic games \emph{in the limit} as $\delta \rightarrow 1$.

\begin{defn}[Weighted Garbling]
	\label{defn:WG}
	A monitoring structure $\Pi=(Y, f)$ is a \emph{weighted garbling} of $\Pi' = (Y', f')$ if, for every $s \in S$, there exist nonnegative weights $\gamma^{t^\prime,y^\prime}_s \geq 0$, $\forall  (t^\prime,y^\prime) \in S \times Y^\prime$, and $\phi_s: S \times Y^\prime \rightarrow \Delta(S \times Y)$ such that for each $a \in A$,
	\begin{equation}
		\label{eq:nnnnn1}
		p(t,y|s,a)= \sum_{(t^\prime,y^\prime) \in S \times Y^\prime}\gamma^{t^\prime,y^\prime}_s \phi_s(t,y|t^\prime,y^\prime)p'(t^\prime,y^\prime|s,a) , \quad \forall (t,y) \in S \times Y.
	\end{equation}
	We say that $\Pi^\prime$ is \emph{more WG-informative} than $\Pi$ if $\Pi$ is a weighted garbling of $\Pi^\prime$. 
\end{defn}

Note that the weights do not depend on action profile $a$. A standard garbling (i.e., joint garbling) is a special type of weighted garbling where all the weights are equal to $1$.

In the following, we also say that $p$ is a joint/weighted garbling of $p'$ when $\Pi$ is a joint/weighted garbling of $\Pi^\prime$, where $p$ and $p'$ are generated from $\Pi$ and $\Pi^\prime$, respectively.

Weighted garbling can be interpreted in various ways. First, we can regard it as a standard garbling with some action-independent deformation of monitoring structure. Note that the expected value of the weights must be $1$ given every $(s,a) \in S \times A$ by definition, i.e., 
$\sum_{(t^\prime,y^\prime) \in S \times Y'} \gamma^{t^\prime,y^\prime}_s p'(t^\prime,y^\prime|s,a) =1$.\footnote{If the matrix of size $\left|A\right| \times (\left|S\right| \times \left|Y'\right|) $, where each row indexed by $a \in A$ corresponds to $p'(\cdot, \cdot|s,a)$, has full-column rank, then $\gamma^{t^\prime,y^\prime}_s = 1, \ \forall (t',y')$ is the only solution that satisfies this condition at $s$. If this is the case for every state $s$, then weighted garbling reduces to joint garbling. Weighted garbling is strictly more permissive than joint garbling only when there are many public signals relative to the number of actions. This is why we need three signals for $\Pi^\prime$ in \autoref{exmp:nnn2}.} 
Hence, we can regard $p^\gamma(t^\prime,y^\prime|s,a) := \gamma^{t^\prime,y^\prime}_s p'(t^\prime,y^\prime|s,a)$ as a proper joint distribution on $S \times Y^\prime$. Thus the weights transform a joint distribution $p^\prime$ into another joint distribution $p^\gamma$ in such a way that $p$ is a standard (joint) garbling of $p^\gamma$.

We can also interpret weighted garbling as \emph{conditional Blackwell informativeness}. Let $\overline{\gamma} \equiv \max_{s,t^\prime,y^\prime}\gamma^{t^\prime,y^\prime}_s$ be the \emph{size} of the weights, and let $\overline{\gamma}^{t^\prime,y^\prime}_s \equiv \gamma^{t^\prime,y^\prime}_s/\overline{\gamma} \in [0,1]$ for each $(t',y')$.
Then, condition \eqref{eq:nnnnn1} becomes: 
\[
\frac{1}{\overline{\gamma}} p(t,y|s,a)= \sum_{(t^\prime,y^\prime) \in S \times Y^\prime}\overline{\gamma}^{t^\prime,y^\prime}_s \phi_s(t,y|t^\prime,y^\prime)p'(t^\prime,y^\prime|s,a).
\]
This is as if $(t,y)$ is generated with probability $\phi_s(t,y|t^\prime,y^\prime)$, conditional on some event that happens with probability $\overline{\gamma}_s^{t',y'}$ when $(t',y')$ realizes at state $s$. Since the average weight is $1$ for each $(s, a) \in S \times A$, the ex-ante probability of this event is $1/\overline{\gamma}$ independent of $(s,a)$, which appears on the left-hand side. Hence, this expression means that $\Pi$ can be interpreted as a joint garbling of some \emph{conditional monitoring structure} that arises from $\Pi^\prime$, conditional on some event occurring with probability $\overline{\gamma}_s^{t',y'}$ given $(t', y')$ at $s$. We can also see this as a two-step garbling of $\Pi^\prime$. Suppose that, when $(t',y')$ occurs, $(t',y')$ is observed with probability $\overline{\gamma}_s^{t',y'}$, but only a null signal $n$ is observed with probability $1-\overline{\gamma}_s^{t',y'}$. This generates a joint distribution $p^n$ on $(S \times Y') \cup \{n\}$. This distribution is a standard garbling of $p^\prime$ since $\overline{\gamma}_s^{t',y'}$ does not depend on $a$. As a second step, we can extend $\phi_s$ to $(S \times Y') \cup \{n\}$ by defining $\phi_s(n|n) =1$ to show that $(1/\overline{\gamma}) p + (1-1/\overline{\gamma})  \delta_n$ is a standard garbling of $p^n$, where $\delta_n$ is the Dirac measure on $n$. Hence, by transitivity, $(1/\overline{\gamma})p + (1-1/\overline{\gamma}) \delta_n$ is a standard garbling of $p^\prime$ when $p$ is a weighted garbling of $p^\prime$ with weight size $\overline{\gamma}$.\footnote{The converse also holds because $p$ is a weighted garbling of $(1/\overline{\gamma}) p + (1-1/\overline{\gamma}) \delta_n$ (with $0$ weight on the null signal) and the weighted garbling order can be shown to be transitive.}

We present two examples of weighted garbling. The first example involves a repeated game, which is a special case of stochastic games (i.e., a game with a single state). For repeated games, both ex-post garbling and joint garbling reduce to the standard garbling, whereas weighted garbling does not.

\begin{example}
	\label{exmp:nnn2}
	
	Consider a simple repeated Prisoners' Dilemma game with two actions $\left\{C,D\right\}$ and the following two monitoring structures $\Pi = (Y, f)$ and $\Pi' = (Y', f')$:
	\begin{itemize}
		\item $\Pi = (Y, f)$: $Y=\{c,d\}$ and  
		$$
		f(c|a_1,a_2)=\begin{cases}
			1-\eta & \text{if $(a_1,a_2)=CC$}\\
			\eta & \text{otherwise}
		\end{cases}
		$$
		\item $\Pi' = (Y', f')$: $Y'=Y \cup\{n\}$ and 
		$$
		f'(c|a_1,a_2)=\begin{cases}
			\epsilon (1-\eta') & \text{if $(a_1,a_2)=CC$}\\
			\epsilon \eta' & \text{otherwise}
		\end{cases}
		$$		
		and 
		$$
		f'(n|a_1,a_2)= 1-\epsilon \ \text{for any $(a_1,a_2)$}
		$$
	\end{itemize}
	where $0< \eta' < \eta < 0.5$ and $\epsilon \in \left(0,1\right]$.

	In the first monitoring structure, the players observe a binary signal that is incorrect with probability $\eta$. In the second monitoring structure, there are three possible signals: the players observe a more informative binary signal with probability $\epsilon$, which is incorrect with a smaller probability $\eta^\prime < \eta$, or they observe no signal (represented by the null signal $n$) with probability $1-\epsilon$.

	These two monitoring structures are not comparable in the sense of Blackwell if $\epsilon$ is not too large: the first monitoring structure cannot be more informative than the second because the posterior beliefs (about actions) can take more extreme values for the second. The second monitoring structure is not more informative than the first either, if $\epsilon$ is small enough.\footnote{The precise condition is $(0.5 - \eta)/(0.5 -\eta') > \epsilon$.} Nonetheless, the second monitoring structure is more WG-informative than the first for any $\epsilon \in \left(0,1\right]$. This can be easily seen from the earlier discussion on interpreting weighted garbling as conditional Blackwell informativeness: the second structure is more informative conditional on observing the more informative binary signal. For instance, we can assign $\gamma^c = \gamma^d = 1/\epsilon, \ \gamma^n = 0$ as weights and set $\phi(c|c) = \phi(d|d) = (1-\eta' -\eta)/(1-2\eta')$ to confirm that $\Pi$ is a weighted garbling of $\Pi^\prime$.
\end{example}

We present the next example to illustrate and compare various notions of garbling.

\begin{example}
	\label{exmp:WG}
	
	Let $S = \{ s_1, s_2\}$. The state transition $q$ is action-independent and is given by $q(t|s) = 1/2$ for each $(t,s) \in S \times S$. Consider the monitoring structures $\Pi = (Y,f) $ and $\Pi' = (Y,f')$, where $Y = \{c,d\}$:
	
	\begin{itemize}
		\item For each $(t,s) \in S^2$, 
		$$
		f(c|t,s, (a_1,a_2))=\begin{cases}
			\frac{2}{3} & \text{if $(a_1,a_2)=CC$}\\
			\frac{1}{3} & \text{otherwise}
		\end{cases}
		$$
		\item For each $s \in S$
		$$
		f' (c|t=s_1,s, (a_1, a_2))=\begin{cases}
			\frac{3}{4} & \text{if $(a_1,a_2)=CC$}\\
			\frac{1}{4} & \text{otherwise}
		\end{cases}
		$$
		and
		$$
		f' (c|t=s_2,s, (a_1, a_2))=\begin{cases}
			\frac{2}{3} - \epsilon & \text{if $(a_1,a_2)=CC$}\\
			\frac{1}{3} + \epsilon & \text{otherwise}
		\end{cases}
		$$
		where $\epsilon \in [0,1/6]$ 
	\end{itemize}
	
	Note that the public signal for $\Pi^\prime$ is strictly more informative than the public signal for $\Pi$ when the next state is $s_1$, but strictly less informative with $\epsilon > 0$ when the next state is $s_2$. The induced joint distributions $p$ and $p'$ from $(q,\Pi)$ and $(q, \Pi')$ are described in \autoref{fig:nnn1}. 
	\begin{figure}
		\center
		\begin{subfigure}[t]{0.4\textwidth}
			\begin{tabular}{|c|c|c|}
				\hline
				$p' (t,y|s,a)$    & $CC$                       & $\neg CC$                  \\ \hline
				$s_1,c$ & $\frac{3}{8}$    & $\frac{1}{8}$     \\ \hline
				$s_1,d$ & $\frac{1}{8}$     & $\frac{3}{8}$     \\ \hline
				$s_2,c$ & $\frac{1}{2}\left(\frac{2}{3} - \epsilon\right)$ & $\frac{1}{2}\left(\frac{1}{3}+ \epsilon\right)$ \\ \hline
				$s_2,d$ & $\frac{1}{2} \left(\frac{1}{3} + \epsilon\right)$ & $\frac{1}{2}\left(\frac{2}{3}- \epsilon\right)$ \\ \hline
			\end{tabular}
		\end{subfigure} 
		\qquad
		\begin{subfigure}[t]{0.4\textwidth}
			\begin{tabular}{|c|c|c|}
				\hline
				$p (t,y|s,a)$    & $CC$                       & $\neg CC$                  \\ \hline
				$s_1,c$ & $\frac{1}{3}$    & $\frac{1}{6}$     \\ \hline
				$s_1,d$ & $\frac{1}{6}$     & $\frac{1}{3}$     \\ \hline
				$s_2,c$ & $\frac{1}{3}$ & $\frac{1}{6}$ \\ \hline
				$s_2,d$ & $\frac{1}{6}$ & $\frac{1}{3}$ \\ \hline
			\end{tabular}
		\end{subfigure}
		\caption{The joint distribution $p'$ and $p$ in \autoref{exmp:WG}}
		\label{fig:nnn1}
	\end{figure}

	We make the following observations:
	\begin{enumerate}
		\item $\Pi$ is an ex-post garbling of $\Pi'$ if and only if $\epsilon = 0$. 
		\item $\Pi$ is a joint garbling of $\Pi'$ if and only if $\epsilon \in [0, 1/12]$.
		\item For any $\epsilon$, $\Pi$ is a weighted garbling of $\Pi'$. 
	\end{enumerate}
	To see the first item, note that if $\epsilon > 0$, then conditional on $t= s_2$, the public signal is strictly more informative for $\Pi$. Hence, $\Pi$ cannot be an ex-post garbling of $\Pi'$. The converse holds trivially by the definition of ex-post garbling. For the second item, it is intuitively clear that $\Pi'$ becomes more informative than $\Pi$ in the standard Blackwell sense if $\epsilon$ is sufficiently small. However, determining the precise upper bound of $\epsilon$ for which this holds requires some work.\footnote{Suppose that $\Pi$ is a joint garbling of $\Pi^\prime$ and let $\phi_s (y|t',y') \equiv \phi_s (s_1, y |t', y') + \phi_s (s_2, y|t',y')$ be the marginal distribution of the public signal given $(t', y')$ for garbling function $\phi_s$. This $\phi_s(y|t',y')$ must satisfy $2/3= 1/2 (3/4\phi_s (c|s_1, c) + 1/4 \phi_s (c|s_1,d) )+ 1/2 ( (2/3 - \epsilon) \phi_s (c|s_2, c) + (1/3+\epsilon) \phi_s (c|s_2, d) )$ for $(a_1, a_2) = CC$ and $1/3 = 1/2 (  1/4\phi_s (c|s_1, c) +  3/4 \phi_s (c|s_1,d) )  +  1/2 ( ( 1/3 + \epsilon )  \phi_s (c|s_2, c) + ( 2/3 - \epsilon ) \phi_s (c|s_2, d)  )$ for $(a_1, a_2) \neq  CC$. Conversely, if there exists $\phi_s(y|t', y')$ that satisfies the above conditions, then we can find a joint distribution $\phi_s (t, y |t', y')$ satisfying the definition of joint garbling, as the next state and public signal are independent for $\Pi$. It can be easily shown that such distribution $\phi_s(y|t', y')$ exists if and only if $\epsilon \in \left[0,1/12 \right]$.}
	To see the third item, note that we can use the following weights: $\gamma_s^{s_1,c} = \gamma_s^{s_1, d} = 2, \gamma_s^{s_2,c} = \gamma_s^{s_2, d} = 0$ for any $s$, and the following garbling: $\phi_{s}(t, c|s_1, c) = \phi_{s}(t, d|s_1, d) = 5/12$ and $\phi_{s}(t, d|s_1, c) = \phi_{s}(t, c|s_1, d) = 1/12$ for any $s, t$.
\end{example}

Observe that a monitoring structure is a greatest element in the WG-order if and only if it is a \emph{conditionally perfect monitoring}.\footnote{Formally, $\Pi = (Y,f)$ is a conditionally perfect monitoring structure if, at each $s \in S$ and $a \in A$, there exists a subset $D_{s,a} \subseteq S \times Y$ such that $p(D_{s,a}|s,a) > 0$ but $p(D_{s,a}|s,a') = 0$ for any $a' \neq a$. This means that the players would know that $a$ was played if they observe any $(t,y) \in D_{s,a}$ at state $s$.} To see this, note that a conditionally perfect monitoring structure is as informative as a perfect monitoring structure in the WG-order.\footnote{A monitoring structure $\Pi = (Y, f)$ is a perfect monitoring structure if, at each $s \in S$, the support of $p(\cdot, \cdot |s,a)$ in $S \times Y$ never overlaps for any two distinct action profiles.} Then, since a monitoring structure is a greatest element in the Blackwell order if and only if it is a perfect monitoring structure, every conditionally perfect monitoring structure must be a greatest element in the WG-order by transitivity.\footnote{It is straightforward to show that the WG-order is transitive.} Conversely, it is clear that any monitoring structure that is not conditionally perfect is strictly dominated by some conditionally perfect monitoring structure.

Our main result in the next section shows that, if $\Pi$ is a weighted garbling of $\Pi'$, then the set of PPE payoffs that can be supported given any initial state is weakly larger with $\Pi^\prime$ in the limit as $\delta \rightarrow 1$ given any payoff functions under some regularity assumptions.

Since joint garbling is a special case of weighted garbling, one corollary of our result is that the monotonicity of PPE payoff set with respect to the standard garbling (i.e., joint garbling) is reestablished in the limit. Recall that this monotonicity does not hold for stochastic games for a fixed discount factor.
\autoref{tab:nnn1} summarizes the relationship between various notions of garbling and the expansion of the PPE payoff set.

\begin{table}[t]
	\center
	\begin{subtable}{0.45\textwidth}
		\center
		\begin{tabular}{|c|c|c|}
			\hline
			& fixed $\delta$ & limit             \\ \hline
			ex-post=joint & o              & o \\ \hline
			weighted        & x              & o* \\ \hline
		\end{tabular}
		\subcaption{Repeated games}
	\end{subtable}
	\qquad 
	\begin{subtable}{0.45\textwidth}
		\center
		\begin{tabular}{|c|c|c|}
			\hline
			& fixed $\delta$ & limit             \\ \hline
			ex-post  & o              & o                 \\ \hline
			joint    & x              & o* \\ \hline
			weighted & x              & o*\\ \hline
		\end{tabular}
		\subcaption{(Irreducible) stochastic games}
	\end{subtable}

	\caption{Table summarizing the expansion of PPE with respect to various notions of garbling for a fixed discount factor (``fixed $\delta$'') or in the limit (``limit'') in repeated games and stochastic games. If an improvement in the monitoring structure in terms of the corresponding notion guarantees the expansion, the relevant box is marked with ``o''; otherwise, it is marked with ``x.'' The contributions of the current paper are indicated with a star.}
	\label{tab:nnn1}
\end{table}

\subsection{Weighted Garbling for Action-Independent Transition}

For applications, it would be useful to have a simple way to rank monitoring structures by weighted garbling. Here, we provide a simple sufficient condition for weighted garbling when the state transition is action-independent. This sufficiency condition generalizes the idea behind \autoref{exmp:WG}.

Let us introduce some new notations to state our condition. For monitoring structure $\Pi = (Y, f)$, we can regard $f( \cdot |t, s, \cdot)$ as a conditional monitoring structure that maps each action profile to a public signal distribution conditional on $(s,t) \in S^2$. We simply refer to this as \emph{$f$ conditional on $(s,t)$}. We can compare such conditional monitoring structures in the WG-order. We say $f^\prime$ conditional on $(s', t')$ for $\Pi^\prime$ is more WG-informative than $f$ conditional on $(s, t)$ for $\Pi$ if there exist a weight $\gamma^{y'} \geq 0$, $\forall y' \in Y'$, and $\phi: Y^\prime \rightarrow \Delta (Y)$ that satisfy $f(y|t, s, a) = \sum_{y^\prime \in Y^\prime} \gamma^{y^\prime}\phi(y | y^\prime)f^\prime(y^\prime|t',s',a)$ for all $y \in Y$ and all $a \in A$.

Let $S(s) \subseteq S$ denote the support of action-independent transition $q( \cdot |s)$ at state $s$. The following proposition provides a simple sufficient condition in terms of $f'$ and $f$ for $\Pi'$ to be more WG-informative than $\Pi$ for the special case of action-independent transitions. 

\begin{prop}
	\label{prop:nnn1}
	Suppose that the transition law is action-independent. For $\Pi = (Y, f)$ and $\Pi^\prime =(Y^\prime, f^\prime)$, suppose that for any $s \in S$, there exists $T_s:S(s) \rightarrow S(s)$ such that $f^\prime$ conditional on $(s, T_s(t))$ is more WG-informative than $f$ conditional on $(s, t)$ for any $t \in S(s)$. Then, $\Pi^\prime$ is more WG-informative than $\Pi$.
\end{prop}

\begin{proof}
	See \Cref{proof:prop1}.
\end{proof}

It may be helpful to compare this to the condition for ex-post garbling. Note that $\Pi$ is an ex-post garbling of $\Pi^\prime$ if and only if $f'$ conditional on $(s,t)$ is more informative than $f$ conditional on the same $(s,t)$ for all $(s, t) \in S^2$. In contrast,  $\Pi$ is a weighted garbling of $\Pi^\prime$ if, for each transition $(s,t)$, there exists some $t'$ such that  $f'$ conditional on $(s,t')$ is more WG-informative than $f$ conditional on $(s,t)$. This condition is satisfied, for example, if the public signal for $\Pi^\prime$ is very informative conditional on staying at the current state, so that $f^\prime$ conditional on $(s,s)$ is more informative than $f$ conditional on $(s,t)$ for any $t \in S$. Note that $f^\prime$ does not need to be informative given any other transition $(s,t)$ with $t \neq s$ in this example. \autoref{exmp:WG} is another example of this proposition, as the public signal is most informative for $\Pi^\prime$ when and only when the next state is $s_1$.

The action-independence of transitions is crucial for this result. If the state transitions depend on actions, it would not be possible to compare two monitoring structures based solely on the public signals, as the next state would also be informative about the current actions.

\section{Weighted Garbling and the Limit PPE Payoff Set}
\label{sec:4}

In this section, we show that an improved monitoring structure in terms of weighted garbling expands the PPE payoff set when players are arbitrarily patient. We then provide an example to illustrate the main ideas. Additionally, we show that the PPE payoff set may expand strictly when the monitoring improves strictly in a certain sense.

\subsection{A Limit Characterization of PPE Payoff Set by HSTV (2011)}
\label{subsec:4.1}

Since we utilize HSTV's characterization of the limit PPE payoff set for stochastic games with imperfect public monitoring, we summarize their relevant results in this subsection.

HSTV considers a collection of programming problems, which generalizes the programming problems introduced by \cite{FL_1994_JET} for repeated games, to study the limit PPE payoff set for stochastic games.

Given a weight $\lambda \in \Lambda \equiv \mathbb{R}^N \setminus \{ \mathbf{0} \}$, consider the following problem $\mathcal{P}(\lambda)$.\footnote{Our notation of $x = \left(x_s, s \in S\right)$ is slightly different from that of HSTV as we use a lower subscript to denote the current state; our $x_{s}(t,y)$ corresponds to $x_{t}(s,y)$ in HSTV.}  \\

$k (\lambda) : = \sup_{v \in \mathbb{R}^N, \ \alpha_{s} \in \prod_{i \in I} \Delta (A_i) \ \forall s, \ x_s(t,y) \in \mathbb{R}^{N} \ \forall s,t,y} \lambda \cdot v$ 
s.t.
\begin{equation}
	\label{eq:1}
	v= u (\alpha_s, s) + \sum_{ (t, y) \in S \times Y} x_s (t,y) p (t,y|s, \alpha_s),\quad \forall s \in S
\end{equation}
\begin{equation}
	\label{eq:2}
	v_i \geq u_i ( (a_i', \alpha_{-i,s}) ,s) + \sum_{(t, y) \in S \times Y} x_{i,s} (t,y) p (t,y|s, (a_i', \alpha_{-i,s})),\quad \forall s \in S,i \in I,  a_i' \in A_i 
\end{equation}
and 
\begin{equation}
	\label{eq:3}
	\lambda \cdot \sum_{ s \in T} x_s (\xi(s),\psi(s)) \leq 0
\end{equation}
for any subset $T \subseteq S$, permutation $\xi:T \to T$, and mapping $\psi: T \to Y$.\footnote{$k(\lambda)$ is finite (see Section 3.4 of HSTV).} That is, given a weight $\lambda \in \Lambda$, $k(\lambda)$ is the maximum state-independent weighted average of players' payoffs $v$ that can be generated for each $s \in S$ with some action profile $\alpha_s$ and some $(t,y)$-contingent ``payment'' $x_s$ that makes $\alpha_s$ incentive compatible (i.e., \eqref{eq:2}) and satisfies some across-states feasibility condition (i.e., \eqref{eq:3}). Each $x_s(t,y)$ represents the (discounted) increments of the total payoffs between the current period and the next period.\footnote{Take any average payoff $v = (1-\delta) u(a,s) + \delta \sum_{(t,y) \in S \times Y}w_s(t,y) p(t,y|s,a)$. With a standard reformulation, this $v$ can be expressed as the sum of the current stage game payoff and the (discounted) difference between the total payoffs from the next period and the total payoff from the current period as follows: $v = u(a,s) + \sum_{(t,y) \in S \times Y}(\delta/(1-\delta))(w_s(t,y)-v)p(t,y|s,a)$. $x_s(t,y)$ corresponds to $(\delta/(1-\delta))(w_s(t,y)-v)$ in the second term of this expression.} Note that this program does not involve $\delta$.

Let $$\mathcal{H} \equiv \bigcap_{ \lambda \in \Lambda} \{ v| \lambda \cdot v \leq k (\lambda)\}.$$ Let $E^\delta (s)$ be the set of PPE payoffs given $\delta$ and initial state $s$, which is compact. HSTV invokes the following assumption to obtain some of their results ((ii) in \autoref{thm:HSTV}).
\begin{assumption}
	\label{ass:1}
	For any $s,s' \in S$, 
	$$\lim_{ \delta \to 1} d (E^\delta (s), E^\delta (s')) = 0,$$
	where $d (E^\delta (s), E^\delta (s'))$ is the Hausdorff distance between $E^\delta (s)$ and $E^\delta (s')$.
\end{assumption}
This assumption is satisfied, for example, when the Markov chain on $S$ is irreducible (or has a unique invariant distribution more generally) given any pure Markov strategy profile.

HSTV proves the following result.
\begin{thm}[HSTV]
	\label{thm:HSTV}
	The following hold:
	\begin{enumerate}[label=(\roman*)]
		\item For each $\delta \in [0,1)$, $\bigcap_{s \in S} E^\delta (s) \subseteq \mathcal{H}.$
		\item Suppose that $\dim (\mathcal{H}) = N$. Then, $\lim_{ \delta \to 1} \bigcap_{s \in S} E^\delta (s) = \mathcal{H}$. In particular, under \autoref{ass:1}, $\lim_{\delta \to 1}E^\delta (s) = \mathcal{H}$ for each $s$.
	\end{enumerate}
\end{thm}
The first item means that $\mathcal{H}$ is an upper bound (in the sense of set inclusion) of the intersection of the PPE payoff sets over initial states for each $\delta$. The second item says that, under the assumption of full-dimensionality, the intersection of the PPE payoff sets converges to $\mathcal{H}$ in Hausdorff distance. In addition, this limit set coincides with the limit PPE payoff set at every state if \autoref{ass:1} is satisfied.\footnote{\autoref{ass:1} and the full-dimensionality assumption together imply the convergence of $E^\delta(s)$ to the same set for every $s$ as $\delta \rightarrow 1$. For an example in which the PPE payoff set does not converge, see \cite{RZ_2020_MOR}.}

Based on this result, HSTV also provides a sufficient condition for a folk theorem, which naturally extends the sufficient conditions for the folk theorem for repeated games with imperfect public monitoring in \cite{FLM_1994_ECMA}.\footnote{HSTV's condition generalizes the individual full rank condition and the pairwise full rank condition in \cite{FLM_1994_ECMA}. See also \cite{FY_2011_JET} for closely related conditions, which are slightly stronger than HSTV's.}

\subsection{Main Result}

We are now ready to state our first main result. In what follows, we use notations such as $\mathcal{P}(\lambda; \Pi), k (\lambda; \Pi)$, $\mathcal{H}(\Pi)$, $E^\delta (s; \Pi)$ to explicitly indicate their dependence on the underlying monitoring structure.  

\begin{thm}
	\label{thm:mono}
	Suppose that $\dim (\mathcal{H} (\Pi) ) = N$ and $\Pi = (Y,f)$ is a weighted garbling of $\Pi' = (Y',f')$. Then, 
	$\lim_{ \delta \to 1} \bigcap_{s \in S} E^{ \delta } (s; \Pi) \subseteq \lim_{ \delta \to 1} \bigcap_{s \in S} E^{ \delta } (s; \Pi^\prime)$. 
	In particular, if \autoref{ass:1} is satisfied for both $\Pi$ and $\Pi^\prime$, then $\lim_{ \delta \to 1} E^{ \delta } (s; \Pi) \subseteq \lim_{ \delta \to 1} E^{ \delta } (s'; \Pi'),  \forall s, s' \in S.$
\end{thm}

In words, the limit PPE payoff set of a stochastic game expands if the monitoring structure improves in terms of weighted garbling, under the full dimensionality assumption. This result has two implications. First, recall that joint garbling is a special case of weighted garbling. Unlike in repeated games, the PPE payoff set for stochastic games does not expand for a fixed $\delta$ when the monitoring structure improves in terms of joint garbling, even though it is just the standard Blackwell garbling. This theorem shows that the monotonicity of the PPE payoff set with respect to joint garbling is restored in the limit. Second, it shows that the limit monotonicity of the PPE payoff set can be achieved with respect to a much weaker garbling notion: weighted garbling. Consequently, we can compare a much larger set of monitoring structures when analyzing the limit equilibrium payoffs. Notably, this second part of the contribution is novel even in repeated games, since weighted garbling does not reduce to the standard garbling in repeated games.

The above result is most useful when the folk theorem does not hold, as there is still room for the limit equilibrium payoff set to expand. We also note that the result holds independent of payoff functions. The limit equilibrium payoff set would expand with a more informative monitoring structure for any payoff functions, as long as the full dimensionality assumption is satisfied.

Note that we do not require any public randomization device for our monotonicity result as the limit PPE payoff set $\mathcal{H}$ is already convex, whereas public randomization is required for the monotonicity of PPE payoff set given a fixed discount factor, because the PPE payoff set at each state may not be convex without a public randomization device with finite states and signals.\footnote{For repeated games with continuously distributed public signal, the PPE payoff set is convex without any public randomization device \cite[]{Kandori_1992_RES}.}

We prove this result by using our next result, which states that the maximum \emph{score} of problem $\mathcal{P}(\lambda)$ increases weakly for every direction $\lambda$, thereby yielding a weakly larger $\mathcal{H}$, when a monitoring structure improves in terms of weighted garbling. Together with \autoref{thm:HSTV}, this immediately implies the above main result as $\mathcal{H}(\Pi')$ is full dimensional when $\mathcal{H} (\Pi)$ is full dimensional.

\begin{thm}
	\label{thm:score}
	Suppose $\Pi = (Y, f)$ is a weighted garbling of $\Pi' = (Y', f')$. Then, for each $\lambda \in \Lambda$, $k (\lambda; \Pi) \leq k (\lambda; \Pi').$ Hence, $\mathcal{H} (\Pi) \subseteq \mathcal{H} (\Pi')$.
\end{thm}

Let us outline the idea of the proof of \autoref{thm:score}. We first observe that the payoff increment $x_s (t,y)$ can be decomposed into two components $l_s (t)$ and $z_s (t,y)$. The first component, $l_s (t)$, is defined as the increment that pushes the score most in the direction of $\lambda$ given the transition $(s,t)$, i.e., $\lambda \cdot l_s (t) = \max_{y \in Y} \lambda \cdot x_s (t,y)$. Note that, by construction, this component satisfies condition \eqref{eq:3} in HSTV's program.

On the other hand, the second component, defined by $z_s (t,y):= x_s (t,y) -l_s (t)$, represents an additional variation of increments based on the public signal. Note that the score of each $z_s (t,y)$ is nonpositive by definition.

For a more WG-informative structure $\Pi^\prime$, we keep the same $l_s (t)$ and construct the weighted expectation of $z_s (t,y)$ given each $(t^\prime,y^\prime)$ as $\tilde z_s (t^\prime,y^\prime) = \sum_{(t,y) \in S \times Y} \gamma^{t^\prime,y^\prime}_s \phi_s(t,y|t^\prime,y^\prime) z_s(t,y)$. This application of garbling to these increments is similar to the application of garbling to continuation payoffs for repeated games (e.g., \cite{Kandori_1992_RES}). We then construct the new payoff increment by adding these two components together, i.e., $\tilde x_s (t^\prime,y^\prime)  = l_s (t^\prime) + \tilde z_s (t^\prime,y^\prime)$. They clearly satisfy conditions \eqref{eq:1} and \eqref{eq:2} for $\mathcal{P}(\lambda; \Pi^\prime)$. This is because the state transition rule is independent of monitoring structures, hence the expectation of $l_s(\cdot)$ is the same for $\Pi$ and $\Pi^\prime$, and the expectation of the second term is the same given the above definition based on weighted garbling. Notice that condition \eqref{eq:3} also holds because the score of each $\tilde z_s (t^\prime, y^\prime)$ is nonpositive by construction. 

\begin{proof}
	
	Fix $\lambda \in \Lambda$. Let $v \in \mathbb{R}, \alpha_s \in \prod_{i \in I} \Delta (A_i)  ,x_s : S \times Y \to \mathbb{R}^N$ for each $s \in S$ be any feasible point for HSTV's program $\mathcal{P}(\lambda; \Pi)$.

	For our purpose, it is sufficient to construct $\{ \tilde{x}_s (t',y')\}_{ s, t' \in S, y' \in Y'}$ that satisfies all the constraints, together with the same $\{ \alpha_s\}_s$ and $v$, for the problem $\mathcal{P}(\lambda; \Pi^\prime)$.

	For each $s,t \in S$, let $l_s (t) = x_s (t,y^*)$, where $y^*$ is a solution to $\max_{y \in Y} \lambda \cdot x_s (t,y)$. We decompose $x_s(t,y)$ into $l_s (t)$ and the remaining part $z_s(t,y)$, where $z_s(t,y) := x_s(t,y)- l_s (t)$.
	
	For each $(t',y') \in S \times Y'$, define
	\[
	\tilde{z}_s(t^\prime, y^\prime) : = \sum_{(t, y) \in S \times Y} \gamma^{t^\prime,y^\prime}_s \phi_s(t,y|t^\prime,y^\prime) z_s(t,y)
	\]
	and 
	\[ 
	\tilde{x}_s (t',y'):=l_s (t') + \tilde{z}_s (t',y').
	\]
	Let $T \subseteq S$. Note that for any permutation $\xi: T \to T$,
	\[ 
	\lambda \cdot \sum_{s \in T} l_s (\xi(s)) = \sum_{s \in T} \max_{y \in Y} \lambda \cdot x_s (\xi (s),y) \leq 0,
	\] 
	where the last inequality follows from \eqref{eq:3}. Moreover, since $\lambda \cdot z_s (t,y) \leq 0$ for all $s,t$ and $y$ by construction, and $\gamma_s^{t',y'} \geq 0$,
	$$\lambda \cdot \tilde{z}_s (t',y')=\sum_{(t, y) \in S \times Y} \gamma_s^{t',y'}\phi_s ( t,y | t',y')\lambda \cdot z_s (t,y)\leq 0,\quad \forall s \in S, t' \in S,y' \in Y.$$
	Then,
	$$\lambda \cdot \sum_{s \in T} \tilde{x}_s (\xi(s), \psi (s))= \lambda \cdot \sum_{s \in T} l_s (\xi(s) ) + \lambda \cdot \sum_{s \in T}  \tilde{z}_s (\xi(s),\psi (s))\leq 0$$
	for any permutation $\xi: T \to T$ and any function $\psi: T \to Y'.$
	Thus, we conclude that condition \eqref{eq:3} is satisfied for $\tilde{x}_s (t',y')$ with the monitoring structure $\Pi'$. To show that the constraints corresponding to \eqref{eq:1} and \eqref{eq:2} are satisfied, fix $s \in S$ and observe that for each $a \in A$,
	\begin{align*}
		\sum_{(t',y') \in S \times Y'} \tilde{z}_s(t',y') p' (t',y' | s,a)&= \sum_{(t',y') \in S \times Y'} \sum_{(t,y) \in S \times Y} \gamma^{t^\prime,y^\prime}_s \phi_s(t,y|t^\prime,y^\prime) z_s(t,y) p' (t',y' |s,a)\\
		&= \sum_{(t,y) \in S \times Y} z_s(t,y) \left(\sum_{(t',y') \in S \times Y'}  \gamma^{t^\prime,y^\prime}_s \phi_s(t,y|t^\prime,y^\prime) p' (t',y' |s,a) \right) \\
		&= \sum_{(t,y) \in S \times Y} z_s(t,y) p (t,y |s,a).
	\end{align*}
	Thus, 
	\begin{align*}
		\sum_{(t', y') \in S \times Y'}  \tilde{x}_s (t',y')p' (t',y' | s,a)&= \sum_{ (t', y') \in S \times Y'} l_s (t') p' (t',y' | s,a) + \sum_{(t', y') \in S \times Y'}\tilde{z}_s (t',y')p' (t',y'|s,a) \nonumber \\
		&=\sum_{t \in S} l_s (t) q (t| s,a) + \sum_{(t, y) \in S \times Y} z_s (t,y)p(t,y | s,a) \label{eq:nnn4}\\
		&= \sum_{(t, y) \in S \times Y} x_s (t,y)p (t,y| s,a). \nonumber 
	\end{align*}
\end{proof}

There are two ideas behind this limit monotonicity result with respect to the WG-order, which does not work for a fixed $\delta$ but works as $\delta \rightarrow 1$. The first idea is to use the next period state as a purely informative signal. Since the next state and public signal are jointly informative about the current action, it is natural to compare the informativeness of monitoring structures with respect to the standard Blackwell garbling order for joint distributions (i.e., joint garbling). However, the next period state is not just an informative signal but also determines the set of feasible continuation payoffs. To ignore such state-dependent physical constraints and use the next state as a purely informative signal, we focus on the set of PPE payoffs that can be supported starting at any state, which converges to $\mathcal{H}(\Pi)$ as $\delta \rightarrow 1$ subject to the full dimensionality as shown by \cite{HSTV_2011_ECMA}. Every such payoff profile (in the interior of $\mathcal{H}(\Pi)$) can be supported by creating a state-independent ``block self-generating set'' around it when $\delta$ is large.\footnote{This is a set that is self-generating independent of the initial state with a block of $n-1$ periods. See the proof of Proposition 2 in \cite{HSTV_2011_ECMA} for details.} If $\Pi^\prime$ is more informative than $\Pi$ with respect to the joint garbling order, then we can mimic the construction of continuation payoffs to support each target payoff profile based on joint garbling. The continuation payoffs we construct are different from the original ones in a complex way in general.\footnote{For repeated games, this reduces to a simple convex combination of continuation payoffs in \Cref{subsec:3.1}.} However, when $\delta$ is very large, they behave similarly to the original continuation payoffs and hence still stay in the original block self-generating set. Thus, the block self-generating set for $\Pi$ becomes block self-generating for $\Pi^\prime$ when the players are patient enough. This is why the set of PPE payoffs that can be supported at every state is monotone with respect to the joint garbling order in the limit as $\delta \rightarrow 1$. In addition, this set is exactly the (identical) limit PPE payoff set at every state when \autoref{ass:1} is satisfied.

The second idea is conditional informativeness. With a more WG-informative monitoring structure, players can wait for more informative signals to enable more efficient rewards or punishments. However, the probability of such an event may be small, hence we may not be able to generate large enough payoff variations to provide enough incentives for a given discount factor. However, this problem is resolved as well when players get arbitrarily patient. Note that this argument is valid even in repeated games. Thus, our result complements the monotonicity result with a fixed $\delta$ for repeated games in \cite{Kandori_1992_RES}.

\subsection{Example}

We present the simplest example with no state (i.e., a repeated game) to illustrate how the score improves with improved monitoring. This example also demonstrates that weighted garbling is a new and useful concept even for repeated games.

Suppose that the stage game is a Prisoners' Dilemma, where the payoff is $(1,1)$ for $(C,C)$, $(1+g, -\ell)$ with $g, \ell > 0$ for $(D,C)$ (symmetric for $(C,D)$), and $(0,0)$ for $(D,D)$ (see \autoref{fig:pd}). Assume $1> g-\ell$, so that $(C,C)$ maximizes the joint payoff.
\begin{figure}[]
	\center
	\begin{tabular}{ccc}
		& $C$                             & $D$                             \\ \cline{2-3} 
		\multicolumn{1}{c|}{$C$} & \multicolumn{1}{c|}{$1,1$}    & \multicolumn{1}{c|}{$-\ell,1+g$} \\ \cline{2-3} 
		\multicolumn{1}{c|}{$D$} & \multicolumn{1}{c|}{$1+g,-\ell$} & \multicolumn{1}{c|}{$0,0$}    \\ \cline{2-3} 
	\end{tabular}
	\caption{Payoff matrix for the Prisoners' Dilemma}
	\label{fig:pd}
\end{figure}
We use the monitoring structure $\Pi = (\{c,d\}, f)$ and $\Pi' = (\{c,d,n\}, f')$ from \autoref{exmp:nnn2} in \Cref{sec:3}.

For repeated games, the programming problem for each $\lambda$ reduces to the standard problem in \cite{FL_1994_JET}. For example, if we set $\lambda = (1,1)$ and set the action profile to $(C, C)$, the problem for $\mathcal{P}(\lambda; \Pi^\prime)$ becomes:
\begin{align*}
	&\sup_{v \in \mathbb{R}^{2} \ x: Y^\prime \to \mathbb{R}^{2}} v_1 + v_2 \ \mbox{s.t.} \\
	&v_i = 1 + \sum_{y' \in Y'} x_i(y') p'(y'|CC),\quad i=1,2 \\
	&v_i \geq 1+g + \sum_{y' \in Y'} x_{i} (y') p'(y'|\neg CC),\quad i=1,2 \\
	&x_{1} (y') + x_2(y') \leq 0, \quad \forall y' \in Y'.
\end{align*}

We can show that $\mathcal{H}(\Pi^\prime)$ is strictly larger than $\mathcal{H}(\Pi)$ in this example for any $\epsilon > 0$ (Recall that $p'$ is implicitly parameterized by the probability $\epsilon$ of observing a more informative binary signal for $\Pi^\prime$). This is easy to see when $\epsilon =1$, i.e., $\Pi^\prime$ is strictly more informative than $\Pi$ in the standard Blackwell sense. In this case, clearly $k(\lambda; \Pi^\prime) \geq k(\lambda; \Pi)$ holds for any $\lambda$. Furthermore, it can be easily shown that the inequality holds strictly for any $\lambda \gg 0$.\footnote{For example, for $\lambda = (1,1)$, the maximized score for $\Pi$ and $\Pi^\prime$ would be $2 - (2 \eta g)/(1-2\eta)$ and $2 - (2 \eta' g)/(1-2\eta')$, respectively. The latter is strictly larger since $\eta' < \eta$.} Then note that $\mathcal{H}(\Pi^\prime)$ is actually the same for any $\epsilon > 0$, i.e., when $\Pi^\prime$ is just a weighted garbling of $\Pi$. If $x_i^*$ is a solution to the above score maximization problem when $\epsilon =1$, then for any $\epsilon \in (0,1)$, we can achieve exactly the same score by setting $x_i(y') =(1/\epsilon) x_i^*(y')$ for $y' =c,d$, and $x_i(n) = 0$.

Essentially, we are replicating the continuation payoff variations after $y=c,d$ for $\Pi$ using $y'=c, d$ for $\Pi^\prime$. The probability of observing $y' = c$ or $d$ is $\epsilon$, which may be too small to satisfy the incentive constraints for a given level of discount factor. However, they are not binding constraints when the discount factor becomes sufficiently large.

\subsection{Strict Weighted Garbling for Strict Improvement}

In this subsection, we examine when the score increases strictly. Note that the maximized score in each direction is only informative about how far the PPE payoff set can expand toward that direction \emph{in the limit}. Hence, our result does not provide a clear comparison of two monitoring structures for $\delta$ less than $1$ when the maximized score is the same for both monitoring structures. For example, $k(\lambda, \Pi^\prime) \geq k(\lambda; \Pi )$ holds for \autoref{exmp:nnn2} for any $\lambda$, even if $\eta^\prime = \eta$. However, $\Pi^\prime$ is clearly less informative than $\Pi$ in this knife-edge case since $\Pi^\prime$ is just a convex combination of $\Pi$ and no information. Hence the PPE payoff set for $\Pi^\prime$ cannot be a strict superset of the PPE payoff set for $\Pi$ for any discount factor $\delta < 1$. What \autoref{thm:mono} says is just that the PPE payoff set for $\Pi^\prime$ must catch up with the PPE payoff set for $\Pi$ at least in the limit. On the other hand, when $\eta^\prime \in (0, \eta)$ for \autoref{exmp:nnn2}, it can be shown that $k(\lambda, \Pi^\prime)$ is strictly larger than $k(\lambda, \Pi)$ for many $\lambda \in \mathbb{R}^2_{++}$ such as $(1,1)$, as noted in the footnote in the previous subsection. So, the PPE payoff set for $\Pi^\prime$ overtakes the PPE payoff set for $\Pi$ in this example when $\delta$ is large but still less than $1$.

For our strict improvement result, we need to introduce a strict version of weighted garbling as follows.   

\begin{defn}[Strict Weighted Garbling]
	A monitoring structure $\Pi=(Y, f)$ is a \emph{strict weighted garbling} of $\Pi' = (Y', f')$ if, for every $s \in S$, there exist nonnegative weights $\gamma^{t^\prime,y^\prime}_{s} \geq 0, \forall (t^\prime,y^\prime)$ and $\phi_{s}: S \times Y^\prime \rightarrow \Delta(S \times Y)$ where the support of $\phi_{s}(\cdot, \cdot |t^\prime,y^\prime)$ in $S \times Y$ is the same for all $(t^\prime,y^\prime) \in S \times Y'$ and, for each $a \in A$,
	\[
	p(t,y|s,a)= \sum_{(t^\prime,y^\prime) \in S \times Y^\prime}\gamma^{t^\prime,y^\prime}_{s} \phi_{s}(t,y|t^\prime,y^\prime)p'(t^\prime,y^\prime|s,a) , \quad \forall (t,y) \in S \times Y.
	\]
\end{defn}

Intuitively, $\Pi$ is a strict weighted garbling of $\Pi^\prime$ when $\Pi$ is a weighted garbling of $\Pi^\prime$ and $\phi$ represents a noise with ``common support'' at each state.\footnote{Recall that when $\Pi$ is a weighted garbling of $\Pi^\prime$ with weight $\gamma$ and $\phi$, $p^\gamma$ is more informative than $p$ in the standard Blackwell sense, where $p^\gamma(t^\prime, y^\prime|s,a) = \gamma_{s}^{t^\prime, y^\prime} p'(t^\prime, y^\prime|s,a)$. Strict weighted garbling corresponds to $p^\gamma$ being more informative than $p$ in strong sense. In fact, this is more than $p^\gamma$ being strictly more Blackwell informative than $p$, as it requires $\phi_s(t, y|t', y')$ to be positive for every $(t,y)$ in $D_s$ and every $(t'y')$.} In \autoref{exmp:nnn2}, $\Pi^\prime$ with $\eta^\prime \in (0, \eta)$ is indeed strictly WG-more informative than $\Pi$, since we can set $\gamma^{c} = \gamma^{d} = 1/\epsilon, \gamma^{n} = 0$, and $\phi(c|c) = \phi(d|d) = (1-\eta -\eta^\prime)/(1-2\eta^\prime) \in (0,1)$.\footnote{$\phi(\cdot|n)$ can be any arbitrary full-support distribution on $\{c,d\}$.}

To improve the score strictly, there must be some room for improvement. For example, for repeated games, no better monitoring structure would be able to expand the limit PPE payoff set if the folk theorem already holds. If we like to improve the score in direction $\lambda$, there must be some score-burning after some realization of the public signal. That is, there must exist some $y \in Y$ such that $p(y | \alpha) > 0$ and $\sum_{i \in I} \lambda_i x_i(y) < 0$ for solution $\alpha$ and $x(\cdot)$ for $\mathcal{P}(\lambda)$. This simply means that the maximum score in direction $\lambda$ requires costly punishment with respect to that direction.

We extend this condition to stochastic games. We say that the public signal is \emph{essential} in direction $\lambda \in \Lambda$ for $\Pi$ if there exists a solution $(\alpha_s, x_s), s \in S$ for the problem $\mathcal{P}(\lambda; \Pi)$ such that, for every $T \subseteq S$ and any permutation $\xi$ on $T$ for which \eqref{eq:3} is binding (i.e., $\lambda \cdot \sum_{s \in T}  \ell_s(\xi(s)) = 0$), there exist $\hat s \in T$, $\hat t \in S$, and $\hat y \in Y$ such that $p(\hat t, \hat y |\hat s, \alpha_{\hat s}) > 0$ and
$\lambda \cdot \ell_{\hat s}(\hat t) = \max_{y \in Y} \lambda \cdot  x_{\hat s}(\hat t, y) > \lambda \cdot x_{\hat s}(\hat t, \hat y)$.
This means that whenever the feasibility constraint \eqref{eq:3} is binding for some cycle of states at the optimal solution, there is some state in the cycle where there is some score-burning after some realization of $(t,y)$.\footnote{Although this is a condition on endogenous variables, it is easy to verify in some special cases. For example, see the partnership application in \Cref{sec:5}, where the state transition is action-independent.} This condition generalizes the above score-burning condition for repeated games. While slightly more complex, it similarly requires that there is a positive chance of costly punishment based on the public signal when the continuation payoffs cycle around the ``efficient'' frontier.

We show that a strictly WG-more informative monitoring structure $\Pi^\prime$ achieves a strictly higher score in direction $\lambda$ when the public signal is essential in direction $\lambda$ for $\Pi$.
\begin{thm}
	\label{thm:strict}
	Suppose $\Pi$ is a strict weighted garbling of $\Pi^\prime =(Y^\prime, f^\prime)$. If the public signal is essential in direction $\lambda \in \Lambda$ for $\Pi$, then $k(\lambda; \Pi^\prime) > k(\lambda; \Pi)$. 
\end{thm}

\begin{proof}
	
	Since $\Pi$ is a strict weighted garbling of $\Pi^\prime$, for each $s \in S$, there exist nonnegative weights $\gamma_s^{t',y'}$ and $\phi_s: S \times Y^\prime \rightarrow \Delta (S \times Y)$ that have the same support in $S \times Y$ given each $(t' , y')$, which we denote by $D_s$, such that, for each $a \in A$, $p(t,y|s,a) = \sum_{(t', y') \in S \times Y^\prime} \gamma_{s}^{t', y'} \phi_{s}(t, y| t', y') p'(t',y'|s,a)$ for all $(t,y) \in S \times Y$.

	Let $(\alpha_s, x_s), s \in S$, be a solution to the problem $\mathcal{P}(\lambda;\Pi)$. Following the proof of \autoref{thm:score}, decompose $x_s(t,y)$ as $x_s(t,y) = \ell_s(t) + z_s(t,y)$, where $\lambda \cdot z_s(t, y) \leq 0, \forall (t,y) \in S \times Y$. For each $s \in S$, define $\tilde z_s$ by $\tilde z_s(t', y') := \sum_{(t,y) \in S \times Y} \gamma_{s}^{t', y'} \phi_{s}(t, y| t', y') z_s(t,y)$ for any $(t', y')$ with strictly positive weight $\gamma_{s}^{t', y'} > 0$ and $\tilde z_s(t', y') = z$ for any $(t', y')$ with $\gamma_{s}^{t', y'} = 0$, where $z \in \mathbb{R}^N$ is any vector such that $\lambda \cdot z < 0$. Then, define $\tilde x_s(t', y') =\ell_s(t') + \tilde z_s(t',y')$ for each $(t', y') \in S \times Y^\prime$. Then, $(\alpha_s, \tilde x_s), s \in S$ is feasible and achieves $k(\lambda; \Pi)$ for the problem $\mathcal{P}(\lambda;\Pi')$, since all the incentive constraints \eqref{eq:2} are preserved and the feasibility constraints \eqref{eq:3} are satisfied for any permutation and mapping, as shown in the proof of \autoref{thm:score}.

	In fact, it can be shown that the feasibility constraint \eqref{eq:3} is slack. Since the public signal is essential in direction $\lambda$ for $\Pi$, for any $T \subseteq S$ and any permutation $\xi$ on $T$ such that \eqref{eq:3} is binding, there exists some $\hat s \in T$ and $(\hat t, \hat y) \in S \times Y$ such that $p(\hat t, \hat y|\hat s, \alpha_{\hat s} ) > 0$ and $\lambda \cdot z_{\hat s}(\hat t, \hat y) < 0$. By the definition of strict weighted garbling, the support of $p(\cdot, \cdot|\hat s, \alpha_{\hat s} )$ is the same as the support of $\phi_{\hat s}(\cdot, \cdot|t', y')$ for each $(t', y')$. So $(\hat t, \hat y) $ is in $D_{\hat s}$, hence $\lambda \cdot \tilde z_{\hat s}(t',y')$ is strictly negative for any $(t', y') \in S \times Y^\prime$ with strictly positive weight $\gamma_s^{t', y'} > 0$. This means that $k(\lambda; \Pi)$ can be achieved for $\mathcal{P}(\lambda; \Pi^\prime)$ while \eqref{eq:3} is slack for any $T \subset S$, any permutation $\xi$ on $T$, and any mapping $\psi:  T \rightarrow Y^\prime$.
	
	Now, modify $\tilde x$ as follows: $\hat x_s(t', y') = \tilde x_s(t', y') + \epsilon \gamma_s^{t', y'} \lambda$ for some $\epsilon > 0$. Note that $(\alpha_s, \hat x_s), s \in S$ is feasible for $\mathcal{P}(\lambda; \Pi^\prime)$ for small enough $\epsilon$, since condition \eqref{eq:3} is slack for $\tilde x$, and \eqref{eq:2} is still satisfied for $\hat x$ (as $\sum_{(t', y') \in S \times Y^\prime}\gamma_s^{t', y'} p'(t', y'| s, a) =1$ for any $(s, a)$).
	
	Moreover, $(\alpha_s, \hat x_s), s \in S$ achieves a strictly higher value than $k(\lambda;\Pi)$, as $\mathbb{E}[\lambda \cdot \hat x_s(\cdot, \cdot)|s, \alpha_s] = \mathbb{E}[\lambda \cdot \tilde x_s(\cdot, \cdot)|s, \alpha_s] + \epsilon \left\|\lambda \right\|^2 > \mathbb{E}[\lambda \cdot \tilde x_s(\cdot, \cdot)|s, \alpha_s] $. Therefore, $k(\lambda;\Pi^\prime) > k(\lambda;\Pi)$. 
\end{proof}

We note that a strict improvement of the maximum score does not necessarily imply a strictly larger limit PPE payoff set. Recall that $\mathcal{H}(\Pi)$ is the intersection of all the half spaces $\{v| \lambda \cdot v \leq k(\lambda; \Pi)\}$ across all directions. Thus, for example, if the hyperplane $\{v| \lambda \cdot v = k(\lambda; \Pi)\}$ does not touch $\mathcal{H}(\Pi)$, then improving the score strictly in direction $\lambda$ does not help $\mathcal{H}(\Pi)$ expand. When $\mathcal{H}(\Pi)$ is a polyhedron, it expands if and only if at least one of its faces moves outward strictly, i.e., $k(\lambda;\Pi^\prime) > k(\lambda;\Pi)$ for some $\lambda$ for which the hyperplane $\{v| \lambda \cdot v = k(\lambda; \Pi)\}$ overlaps with a face of $\mathcal{H}(\Pi)$.

We also note that a strict expansion is relatively easier to verify for the limit \textit{strongly symmetric equilibrium (SSE)} payoff set. In the next section, we show that the limit SSE payoff set is characterized by the maximum scores in only two directions: $\lambda =(+1,\dots,+1)$ and $\lambda =(-1,\dots,-1)$. If the maximum score improves weakly in these two directions and strictly in at least one of them, then the limit SSE payoff set would become strictly larger. 

\section{Monotonicity of the Limit SSE Payoff Set}
\label{sec:5}

\subsection{Strongly Symmetric Equilibrium for Symmetric Stochastic Games}

We extend our results to an important class of PPE for symmetric games: \emph{strongly symmetric equilibrium (SSE)}. There are several reasons for pursuing this extension. First, SSE is a natural refinement of PPE in symmetric environments. Second, applications often focus on SSE, as the SSE payoff set can be characterized and computed more easily. In fact, we show that the set of SSE payoffs is characterized by two programming problems, rather than an infinite number of programming problems in all directions of $\lambda \in \Lambda$.

We note that our WG-order and monotonicity result is particularly useful for SSE. In symmetric games, the first-best symmetric payoff is rarely achievable in SSE. Since every player is treated symmetrically, the incentive for cooperation needs to be provided through costly group punishment in SSE, which leads to a failure of the folk theorem. This leaves a room for an improved monitoring structure to reduce the incentive cost and enhance efficiency.

We introduce symmetric stochastic games. A stochastic game is \emph{symmetric} if it satisfies the following conditions: 1) symmetric action sets, i.e., $A_i = B$ for some finite set $B$ for every $i \in I$, 2) symmetric payoff functions, i.e., $u_{i}(a_{\xi(1)},\dots,a_{\xi(N)}, s) = u_{\xi(i)}(a,s) $ for any $(a, s) \in B^N \times S$ and any permutation $\xi$ on $I$, 3) symmetric transition law, i.e., $q(\cdot|s, a)= q(\cdot |s, a')$ for any $(s, a)$ and $a'$ that is a permutation of $a$, and 4) symmetric monitoring structure, i.e., $f(\cdot|t,s,a)= f(\cdot |t,s,a')$ for any $(t, s), a$ and $a'$ that is a permutation of $a$. Note that conditions 3) and 4) imply $p (\cdot |s,a) = p (\cdot |s,a')$ for such $a$ and $a'$ for each $s$. From now on, we drop the subscript $i$ and use $u(a, s)$ to denote each player's payoff when every player plays $a \in B$ at state $s$, and $u(a^\prime/a, s)$ to denote the payoff of playing $a^\prime \in B$ when every other player plays $a$ at state $s$. We also interpret distributions such as $q(\cdot|s, a)$ as the distribution when every player plays $a$, and $q(\cdot|s, a^\prime/a)$ as the distribution when one player unilaterally deviates to $a^\prime$ while every other player plays $a \in B$.

A strategy profile is a strongly symmetric equilibrium (SSE) if it is a PPE and, at any history, all players play the same mixed action in equilibrium. Denote the set of SSE payoffs at state $s$ with monitoring structure $\Pi$ by $E^{\delta}_{SS}(s; \Pi)$.

We pay special attention to pure strategy strongly symmetric equilibrium (PSSE), as we often assume pure strategies for applications (we provide one such application in \Cref{subsec:partner}). The set of PSSE payoffs is denoted by $E^{\delta}_{PSS}(s; \Pi)$.

\subsection{$P$-weighted Garbling for Pure Strategy Strongly Symmetric Equilibrium}

For pure strongly symmetric strategies, the quality of a monitoring structure is determined by its informativeness regarding unilateral deviations from symmetric action profiles. Thus, it is natural to introduce a weaker criterion of informativeness when we restrict our attention to PSSE. For any $a \in A$, let $A(a) \subseteq A$ be the set of action profiles consisting of $a$ itself and the action profiles that would arise from all possible unilateral deviations from $a$. We define \emph{$P$-weighted garbling} with respect to a subset of action profiles $A' \subseteq A$ as follows.

\begin{defn}[$P$-Weighted Garbling with respect to $A' \subseteq A$]
	A monitoring structure $\Pi=(Y, f)$ is a \emph{P-weighted garbling} of $\Pi' = (Y', f')$ with respect to $A' \subseteq A$ if, for every $s \in S$ and $a \in A'$, there exist nonnegative weights $\gamma^{t^\prime,y^\prime}_{s,a} \geq 0, \ \forall (t^\prime,y^\prime)$ and $\phi_{s,a}: S \times Y^\prime \rightarrow \Delta(S \times Y)$ such that, for each $\hat a \in A(a)$,
	\[
	p(t,y|s, \hat a)= \sum_{(t^\prime,y^\prime) \in S \times Y^\prime}\gamma^{t^\prime,y^\prime}_{s,a} \phi_{s,a}(t,y|t^\prime,y^\prime)p'(t^\prime,y^\prime|s, \hat a) , \quad \forall (t,y) \in S \times Y.
	\]
	$\Pi$ is a \emph{strict $P$-weighted garbling} of $\Pi'$ with respect to $A' \subseteq A$ if, in addition, $\phi_{s,a}(\cdot, \cdot|t^\prime,y^\prime)$ has the same support $D_{s,a} \subseteq S \times Y$ for every $(t^\prime,y^\prime) \in S \times Y'$ for each $(s, a) \in S \times A'$.	
\end{defn}

This is weaker than weighted garbling, as we restrict attention to two types of subsets of actions: $A'$ and $A(a)$ for each $a \in A'$. $A'$ is a subset of action profiles that we care about for our analysis. For PSSE, this corresponds to the set of all symmetric pure action profiles. $A(a)$ is the set of action profiles that could arise after a unilateral deviation when $a \in A'$ is supposed to be played. Note that, for $\Pi^\prime$ to be more informative than $\Pi$ in terms of PSSE, it only needs to be more informative with respect to $A(a)$ for each symmetric action profile $a$. Also note that weight $\gamma^{t^\prime,y^\prime}_{s,a}$ and $\phi_{s,a}$ depend on $a \in A'$, in addition to the current state $s$. Intuitively, the target action profile serves as an additional state, representing the range of action profiles that are relevant to our analysis.

\subsection{Results for SSE and PSSE}

To extend our result to SSE, we need to extend HSTV's limit characterization to SSE. Note that HSTV's results do not directly apply to the limit SSE payoff set, as the full-dimensionality assumption does not hold for the SSE payoff set.

We show that to characterize the limit SSE payoff set, it suffices to solve the following programming problem $\mathcal{P}_{SS}(\lambda; \Pi)$ for $\lambda = -1$ and $+1$: 
$$ k_{SS}(\lambda; \Pi) =\sup_{v \in \mathbb{R}, \ \alpha_s \in \Delta (B) \ \forall s, \ x_s (t,y) \in \mathbb{R} \ \forall s,t,y } 
\lambda v $$
subject to 
\begin{align*}
	&v= u(\alpha_s, s)   + \sum_{(t, y) \in S \times Y} x_{s} (t,y) p(t,y | s, \alpha_s), \quad \forall s \in S \\
	&v \geq u(a^\prime/\alpha_s, s) + \sum_{(t, y) \in S \times Y} x_{s} (t,y) p (t,y| s, a^\prime/\alpha_s), \quad \forall s \in S, a' \in B\\
	&\lambda \sum_{ s \in T} x_s(\xi (s) ,\psi (s))\leq 0
\end{align*}
for any $T \subseteq S$, any permutation $\xi: T \to T$ and any function $\psi: T \to Y$. Note that $v$ and $x_s(t, y)$ are just numbers this time, because all players receive the same payoff. When we restrict attention to pure strongly symmetric strategies, we replace $\Delta (B)$ with $B$ and refer to this modified problem by $\mathcal{P}_{PSS}(\lambda; \Pi)$. The maximized score for $\mathcal{P}_{PSS}(\lambda; \Pi)$ is denoted by $k_{PSS}(\lambda; \Pi)$.

The following proposition provides an extension of HSTV's limit characterization result to SSE and PSSE.

\begin{prop}
	\label{prop:sn2}
	The following hold:
	\begin{enumerate}
		\item For any $\delta \in [0,1)$, $$\bigcap_{ s \in S} E^\delta_{SS} (s; \Pi)  \subseteq \left[-k_{SS}(-1; \Pi), k_{SS}(+1; \Pi)\right].\footnote{If $k_{SS}(+1; \Pi) < -k_{SS}(-1; \Pi)$, then the intersection of the equilibrium payoff sets would be empty.}$$
		If $-k_{SS}(-1; \Pi) < k_{SS}(+1; \Pi)$, then $\lim_{\delta \to 1} \bigcap_{  s \in S} E^\delta_{SS}(s; \Pi) = \left[-k_{SS}(-1; \Pi), k_{SS}(+1; \Pi)\right].$
		
		\item For any $\delta \in [0,1)$, $$\bigcap_{ s \in S} E^\delta_{PSS} (s; \Pi)  \subseteq \left[-k_{PSS}(-1; \Pi), k_{PSS}(+1; \Pi)\right].$$ If $-k_{PSS}(-1; \Pi) < k_{PSS}(+1; \Pi)$, then $\lim_{\delta \to 1} \bigcap_{  s \in S} E^\delta_{PSS}(s; \Pi) = \left[-k_{PSS}(-1; \Pi), k_{PSS}(+1; \Pi)\right].$ 
	\end{enumerate}
	
\end{prop}

\begin{proof}
	See \Cref{proof:sn2}.
\end{proof}

Using this result, we can extend our main monotonicity results (\autoref{thm:mono}, \autoref{thm:score}, and \autoref{thm:strict}) to SSE and PSSE immediately. To state the next results, we use a variation of \autoref{ass:1} for the limit SSE and PSSE payoff set. Assumption SS (PSS) is satisfied if the distance between the SSE (PSSE) payoff sets at any two states converges to $0$ as $\delta \rightarrow 1$.\footnote{These assumptions are satisfied, for example, when the Markov chain given any symmetric pure Markov strategy is irreducible (or has a unique invariant distribution more generally). Note that they are weaker than the corresponding sufficient condition for the original \autoref{ass:1}, as we restrict attention to strongly symmetric strategies.} We omit the proof of the following proposition as it is straightforward. 

\begin{prop}
	\label{prop:sn3}
	The following hold:
	\begin{enumerate}
		\item Suppose that $-k_{SS}(-1; \Pi) < k_{SS}(+1; \Pi)$ and $\Pi$ is a weighted garbling of $\Pi'$. Then, $\lim_{ \delta \to 1}\bigcap_{ s \in S} E^\delta_{SS} (s; \Pi) \subseteq \lim_{ \delta \to 1} \bigcap_{ s \in S} E^\delta_{SS} (s; \Pi')$. In particular, if Assumption SS is satisfied for $\Pi$ and $\Pi^\prime$, then $\lim_{ \delta \to 1} E^\delta_{SS} (s; \Pi) \subseteq \lim_{ \delta \to 1} E^\delta_{SS} (s'; \Pi')$ for any $s, s' \in S$.
		If, in addition, $\Pi$ is a strict weighted garbling of $\Pi'$ and the public signal is essential in direction $\lambda = -1$ or $+1$ for $\mathcal{P}_{SS}(\lambda; \Pi)$, then the above inclusions are strict. 
		
		\item Suppose that $-k_{PSS}(-1; \Pi) < k_{PSS}(+1; \Pi)$ and $\Pi$ is a $P$-weighted garbling of $\Pi'$ with respect to the set of symmetric action profiles. Then, $\lim_{ \delta \to 1}\bigcap_{ s \in S} E^\delta_{PSS} (s; \Pi) \subseteq \lim_{ \delta \to 1} \bigcap_{ s \in S} E^\delta_{PSS} (s; \Pi')$. In particular, if Assumption PSS is satisfied for $\Pi$ and $\Pi^\prime$, then $\lim_{ \delta \to 1} E^\delta_{PSS} (s; \Pi) \subseteq \lim_{ \delta \to 1} E^\delta_{PSS} (s'; \Pi')$ for any $s, s' \in S$.
		If, in addition, $\Pi$ is a strict $P$-weighted garbling of $\Pi'$ and the public signal is essential in direction $\lambda = -1$ or $+1$ for $\mathcal{P}_{PSS}(\lambda; \Pi)$, then the above inclusions are strict.

	\end{enumerate}
\end{prop}

\subsection{Application: Partnership Game}
\label{subsec:partner}

We apply the above results to PSSE in a symmetric partnership game. Consider $N \in \mathbb{N}$ agents who simultaneously choose either $e$ (``effort'') or $ne$ (``no effort''). Choosing $ne$ incurs no cost, while $e$ costs $c > 0$ for each player. Since the game is symmetric, we use the total number of efforts $ k \in \{0,\dots,N\}$ instead of the whole effort profile to simplify the exposition in the following. There is state $s \in S$ that affects the productivity of the team and evolves over time. We assume that state transition is action-independent and follows the transition law $q: S \rightarrow \Delta(S)$.

Let $\pi(k,s) \in \mathbb{R}$ represent the expected revenue given $k$ efforts at state $s$. The $N$ agents share the revenue equally, so agent $i$'s expected payoff is given by $u(k,s) -c = \frac{1}{N}\pi(k,s)  -c$ when she is one of the $k$ agents who exert effort, and $u(k,s)$ if she shirks while $k$ other agents exert effort.

In each period, the team observes a public signal $y \in Y$ that is informative about the total effort in the current period, which is distributed according to $f(\cdot|t, s, k) \in \Delta(Y)$. Note that one agent's deviation to shirking is indistinguishable from any other agent's deviation to shirking.

We focus on pure strategy SSE.\footnote{If $u(k,s)$ and $f(\cdot|t,s,k)$ are linear in $k$ at each $s$, it is without loss of generality to restrict attention to pure strategies to obtain the maximum score for $\mathcal{P}(\lambda; \Pi)$ in any direction $\lambda$. See Lemma 4.1 of \cite{FL_1994_JET} for a related result.} Thus, we compare the value of monitoring structures based on $P$-weighted garbling with respect to two symmetric action profiles: full-effort and no-effort. Hence, in the following, we use the total number of effort $0$ and $N$ to represent the no-effort profile and the full-effort profile respectively and denote $A' = \{0, N\}$.

As mentioned, we apply \autoref{thm:strict} to PSSE in this partnership game. As a first step, we show that strict $P$-weighted garbling with respect to $\{0, N\}$ is equivalent to certain likelihood ratio conditions in this setting. For monitoring structure $\Pi$, define the likelihood ratios with respect to $k$ and $k-1$ efforts as follows: $L^\Pi_{s,k}(t,y) := p(t, y|s, k)/p(t, y|s, k-1) = f(y|t,s,k)/f(y|t,s,k-1)$ for each $(t,y) \in S \times Y$ at each $s \in S$. Each likelihood can take an infinite value and is set to be $1$ if it is $\frac{0}{0}$. Let $\overline{L}^\Pi_{s,k} = L^\Pi_{s,k}(\overline{t}_{s,k},\overline{y}_{s,k}) = \max_{(t,y) \in S \times Y} L^\Pi_{s,k}(t,y)$ and $\underline{L}^\Pi_{s,k} = L^\Pi_{s,k}(\underline{t}_{s,k},\underline{y}_{s,k}) = \min_{(t,y) \in S \times Y}L^\Pi_{s,k}(t,y)$. For $P$-weighted garbling with respect to symmetric action profiles, the relevant maximum and minimum likelihood ratios are $\overline{L}^\Pi_{s,k}$ and $\underline{L}^\Pi_{s,k}$ for $k=1,N$. We say that $\Pi = (Y, f)$ is \emph{informative} at symmetric action profiles $k=0, N$ if $\underline{L}^\Pi_{s,k} < \overline{L}^\Pi_{s,k}$ for every $s \in S$ and $k=1, N$. This means that a unilateral deviation from any symmetric action profile must be detectable.

The following lemma shows that $\Pi^\prime$ is informative and more informative than $\Pi$ in the strict $P$-weighted garbling order with respect to $\{0, N\}$ if and only if these likelihood ratios for $k=1,N$ take strictly more extreme values for $\Pi^\prime$.

\begin{lem}
	\label{strictWG}
	
	$\Pi$ is a strict $P$-weighted garbling of $\Pi^\prime$ with respect to $\{0, N\}$ and $\Pi^\prime$ is informative at $\{0,N\}$ if and only if $\left[\underline{L}^\Pi_{s,k}, \overline{L}^\Pi_{s,k} \right] \subset \left(\underline{L}^{\Pi^\prime}_{s,k}, \overline{L}^{\Pi^\prime}_{s,k} \right)$ for every $s \in S$ and $k \in \{1, N\}$.
\end{lem}

\begin{proof}
	See \Cref{proof:strictWG}.
\end{proof}

In this setting, strict $P$-weighted garbling is characterized by a simple condition: a more informative monitoring structure has a larger maximum likelihood ratio and a smaller minimum likelihood ratio under both full effort and no effort. It is useful to compare this property to the corresponding property of Blackwell informativeness. Note that $A(N)$ and $A(0)$ essentially contain only two elements: the symmetric action profile and a unilateral deviation from it. Thus, comparing two monitoring structures at these symmetric action profiles parallels the comparison of two experiments with binary states. In this special case, there is a one-to-one relationship between posterior beliefs and likelihood ratios. Consequently, a larger maximum likelihood ratio and smaller minimum likelihood ratio are equivalent to a larger maximum posterior belief and a smaller minimum posterior belief (given any full-support prior). Recall that an experiment is more Blackwell informative than another if and only if the distribution of posterior beliefs for the former experiment is second-order stochastically dominated by the distribution for the latter. In contrast, $P$-weighted garbling in the symmetric partnership game only requires that the largest and smallest posterior beliefs are more extreme (and strictly so for strict $P$-weighted garbling).

Suppose that the above likelihood ratio conditions are satisfied for $\Pi$ and $\Pi^\prime$. By \autoref{strictWG}, $\Pi$ is a strict $P$-weighted garbling of $\Pi^\prime$ with respect to $\{0, N\}$. By \autoref{prop:sn3}, the best limit PSSE payoff improves strictly for $\Pi'$ when the above likelihood ratio conditions are satisfied and the public signal is essential in direction $\lambda = +1$. Since the state transition is action-independent, the public signal must be essential in direction $\lambda = +1$ when $k_{PSS}(+1, \Pi)$ is achieved with symmetric effort profiles that are not stage-game Nash equilibria. Therefore, we obtain the following proposition. 

\begin{prop} Suppose that (1) $\left[\underline{L}^\Pi_{s,k}, \overline{L}^\Pi_{s,k} \right] \subset \left(\underline{L}^{\Pi^\prime}_{s,k}, \overline{L}^{\Pi^\prime}_{s,k} \right)$ holds for every $s \in S$ and $k \in \{0,N\}$ for $\Pi$ and $\Pi^\prime$ and (2) $k_{PSS}(+1; \Pi)$ is achieved with symmetric effort profiles $\left(k_s\right)_{s \in S}$ where $k_s \in \{0,N\}$ is not a stage-game Nash equilibrium at each state $s \in S$. Then, $k_{PSS}(+1; \Pi^\prime) > k_{PSS}(+1; \Pi)$. In particular, if $k_{PSS}(+1; \Pi) > - k_{PSS}(-1; \Pi)$, then, for any high enough $\delta$, there exists a PSSE for $\Pi^\prime$ that generates a strictly larger payoff than any PSSE for $\Pi$.
\end{prop}

\section{Weighted Garbling with Different Transition Laws}
\label{sec:6}
Thus far, we have fixed a transition law when comparing monitoring structures, as a transition law affects not only the informational environment but also the physical environment of stochastic games. However, comparing two monitoring structures under different transition laws may still be meaningful if those transition laws induce the same limit feasible payoff set at each state under each Markov strategy profile, as our focus is on limit equilibrium payoffs. This is the case, for example, if two transition laws have the same unique invariant distribution under any pure Markov strategy profile given any initial state.\footnote{This is because every extreme point of the feasible payoff set at any state can be generated by some pure Markov strategy profile. Hence, we can generate any feasible payoff by randomizing over pure Markov strategy profiles \cite[]{Dutta_1995_JET}.} 

To fix ideas, consider the following example in which the transition laws are action-independent.

\begin{example}
	\label{exmp:nnnnn1}
	
	Let $ S = \{ s_1, s_2\}$ and consider transition laws $q$ and $q'$ such that $q (s_j |s_j) = \eta$ (resp. $q' (s_j |s_j ) = \eta')$ for each $j =1,2$, where $1> \eta' > \eta \geq 0$. This means that the state is more persistent with $q'$. 
	
	Consider the monitoring structures $\Pi = (Y, f) $ and $\Pi' = (Y, f')$ with $Y = \{ c,d\}$, where $f ( c |s_j,s_k, CC ) =f(d|s_j, s_k, \neg CC) = 2/3$ and $f' ( c |s_j,s_k, CC ) =f'(d|s_j, s_k, \neg CC) = 3/4$ for $\Pi$ and $\Pi'$ respectively if $j \neq k$; otherwise, it is $1/2$ independent of action profiles for both $\Pi$ and $\Pi'$.  Thus, the public signal is informative only if the next-period state differs from the current one. The pair $(q,\Pi)$ (resp. $(q',\Pi')$ ) induces the joint distribution $p$ (resp. $p'$) over $S \times Y$ (see \autoref{fig:nnnnn1}).

	\begin{figure}[]
		\center
		\begin{tabular}{|c|c|c|}
			\hline
			$p (t,y |s_j,a)$ & $CC$                   & $\neg CC$              \\ \hline
			$s_j, c$                           & $\eta \frac{1}{2}$     & $\eta \frac{1}{2}$     \\ \hline
			$s_j, d$                           & $\eta \frac{1}{2}$     & $\eta \frac{1}{2}$     \\ \hline
			$s_k, c$                           & $(1-\eta) \frac{2}{3}$ & $(1-\eta) \frac{1}{3}$ \\ \hline
			$s_k,d$                            & $(1-\eta) \frac{1}{3}$ & $(1-\eta)\frac{2}{3}$  \\ \hline
		\end{tabular}
		\quad 
		\begin{tabular}{|c|c|c|}
			\hline
			$p' (t,y |s_j,a)$ & $CC$                   & $\neg CC$              \\ \hline
			$s_j, c$                           & $\eta' \frac{1}{2}$     & $\eta' \frac{1}{2}$     \\ \hline
			$s_j, d$                           & $\eta' \frac{1}{2}$     & $\eta' \frac{1}{2}$     \\ \hline
			$s_k, c$                           & $(1-\eta') \frac{3}{4}$ & $(1-\eta') \frac{1}{4}$ \\ \hline
			$s_k,d$                            & $(1-\eta') \frac{1}{4}$ & $(1-\eta')\frac{3}{4}$  \\ \hline
		\end{tabular}
		\caption{$j = 1, 2$ and $k \neq j$. The joint distribution $p$ (resp. $p'$) induced by $q$ and $\Pi$ (resp. $q'$ and $\Pi'$) in \autoref{exmp:nnnnn1}.}
		\label{fig:nnnnn1}
	\end{figure}

	Note that since $q$ and $q'$ share the same invariant distribution (which assigns equal probability to $s_1$ and $s_2$), they yield the same limit feasible payoff set for any payoff function. Also note that, although players find the public signal informative more frequently with $q$ than with $q'$ (because $\eta' > \eta$), conditional on the transition to a different state, the public signal is more informative with $q'$ (because $3/4 > 2/3$). Together, the results in the previous sections suggest that the limit equilibrium payoff set would be larger with $(q', \Pi')$. The main result of this section confirms that this is indeed the case.

\end{example}

To state the main result of this section, let us first introduce an extension of weighted garbling to accommodate a comparison of monitoring structures with different transition laws: Let $q$ and $q'$ be transition laws, and $\Pi = (Y,f)$ and $\Pi' = (Y',f')$ be monitoring structures. Let $p$ and $p'$ be the joint distributions induced by $(q, \Pi)$ and $(q', \Pi^\prime)$, respectively. We say that $(q,\Pi)$ is a ``weighted garbling'' of $(q',\Pi')$ if condition \eqref{eq:nnnnn1} in the original definition (\autoref{defn:WG}) is satisfied for $p$ and $p'$. For instance, in \autoref{exmp:nnnnn1}, $(q,\Pi)$ is a ``weighted garbling'' of $(q',\Pi')$ according to this definition, because, conditional on the state transitioning to the opposite state, the public signal is more informative for $\Pi'$.

\begin{thm}
	\label{thm:invariant}
	Let $q$ and $q'$ be transition laws that induce irreducible Markov chains with the same invariant distribution for any pure Markov strategy profile. Suppose $(q,\Pi)$ is a ``weighted garbling'' of $(q',\Pi')$ for monitoring structures $\Pi= (Y,f)$ and $\Pi' = (Y',f')$. Then for each direction $\lambda \in \Lambda$, $k(\lambda; (q',\Pi')) \geq k (\lambda; (q,\Pi))$.
\end{thm}

\begin{proof}
	See \Cref{proof:them4}.
\end{proof}

As in the results of the previous sections, the monotonicity of the scores in each direction leads to an expansion of the limit equilibrium payoff set, given the full-dimensionality assumption and \autoref{ass:1}.

This result implies that the ``persistency'' of the state does not affect the limit equilibrium payoffs.\footnote{In stochastic zero-sum games where only one of the players observes the state, \cite{PT_2017_ECMA} compare transition laws and show that the value of the informed player decreases in a notion of persistency (see also \cite{HRSV_2010_OR}).} Suppose that, given any current state $s$ and any current action profile $a$, the next period state follows a transition law $q(\cdot|s,a)$ with probability $\alpha \in (0,1]$, and remains the same with probability $1-\alpha$. Then, the limit equilibrium payoff set is independent of $\alpha \in (0,1]$ because the invariant distribution for $\alpha q(\tilde s |s, a ) + (1-\alpha) \mathbf{1} (\tilde s =s)$ is independent of $\alpha > 0$.

We note that our previous approach of decomposing payoff increments into the physical parts ($(l_s (t))_{s,t \in S}$ in the proof) and the informational parts does not extend to this case. This approach relies on the property that the expected values of the physical parts are unaffected by different monitoring structures. However, this property no longer holds for two different transition laws, even if they generate the same invariant distribution for each pure Markov strategy profile.\footnote{Action-independent transition laws are exceptions.} For this reason, we adopt an alternative approach of employing the dual problem of HSTV studied by \cite{HTV_2014_GEB} for the proof of this result.

\section{Conclusion}
\label{sec:7}
In this paper, we introduce a novel information order on monitoring structures for stochastic games with imperfect public monitoring based on weighted garbling, which extends Blackwell garbling.

We demonstrate that the limit PPE payoff set is monotone with respect to this information order. More specifically, the value of each score maximization problem increases (strictly) as the monitoring structure becomes (strictly) more WG-informative.

We further establish that the monotonicity of the limit SSE payoff set under a version of the WG-order adapted to symmetric environments. This follows from the observation that the limit SSE payoff set is characterized by a subset of the same score maximization problems.

For future work, it may be fruitful to explore other dynamic game environments in which the limit equilibrium payoff set is characterized by a collection of optimization problems. For example, one promising direction is to examine the value of monitoring structures in relation to the set of Belief-free equilibria \cite[]{EHO_2005_ECMA, Yamamoto_2009_JET} for repeated or stochastic games with private monitoring.

\newpage

\appendix
\section{Appendix: Omitted Proofs}
\label{appendix}
\subsection{Proof of \autoref{prop:nnn1}}
\label{proof:prop1}

\begin{proof}
	
	Fix $s \in S$. By assumption, we can define $T_s: S(s) \rightarrow S(s)$ such that, for each $t \in S(s)$, there exist $\gamma_{st}^{y^\prime} \geq 0$ and $\phi_{st}(\cdot|y^\prime) \in \Delta(Y), \forall y^\prime \in Y^\prime$ that satisfy $f(y|s,t,a) =\sum_{y^\prime \in Y'} \gamma_{st}^{y^\prime} \phi_{st}(y|y^\prime) f^\prime(y^\prime|s, T_s(t), a)$ for all $y \in Y $ for each $a \in A$, i.e., $f^\prime$ conditional on $(s, T_s(t))$ is more WG-informative than $f$ conditional on $(s, t)$.

	Define $\gamma_s^{t^\prime, y^\prime}$ for each $(t^\prime, y^\prime)$ by $\gamma_s^{t^\prime, y^\prime} = (\sum_{\tilde t: T_s(\tilde t) = t^\prime }  \gamma_{s\tilde t}^{y^\prime} q(\tilde t|s))/q(t^\prime|s)$ if $t^\prime \in S(s)$, and $\gamma_s^{t^\prime, y^\prime} = 0$ otherwise. Next, for each $(t^\prime, y^\prime) \in S(s) \times Y^\prime$ such that $T^{-1}_s(t^\prime) \neq \emptyset$, define $\phi_s(t,y|t^\prime, y^\prime)$ by $\phi_s(t, y|t^\prime,y^\prime ) = (\gamma_{st}^{y^\prime} \phi_{st}(y|y^\prime) q(t|s))/(\sum_{\tilde t: T_s(\tilde t) = t^\prime} \gamma_{s\tilde t}^{y^\prime} q(\tilde t|s))$ if $t \in T_s^{-1}(t^\prime)$, and $0$ otherwise. For any other $(t^\prime, y^\prime) \in S \times Y^\prime$, let $\phi_s(\cdot, \cdot|t^\prime, y^\prime )$ be an arbitrary distribution on $S \times Y$. Note that, for each $(t^\prime, y^\prime) \in S(s) \times Y^\prime$ such that $T^{-1}_s(t^\prime) \neq \emptyset$, $\sum_{(t,y) \in S \times Y} \phi_s(t,y|t^\prime, y^\prime) = \sum_{t: T_s(t) = t^\prime} ( \gamma_{st}^{y^\prime} q(t|s)/(\sum_{\tilde t: T_s(\tilde t) = t^\prime} \gamma_{s\tilde t}^{y^\prime} q(\tilde t|s))) =1$.

	Observe that for any $(t, y) \in S(s) \times Y$, 
	\begin{align*}
		&\sum_{(t^\prime , y^\prime) \in S \times Y' } \gamma_s^{t^\prime, y^\prime}\phi_s(t,y|t^\prime, y^\prime) p^\prime(t^\prime, y^\prime|s, a)\\
		&\quad=\sum_{y^\prime \in Y' } \gamma_s^{T_s(t), y^\prime}\phi_s(t,y|T_s(t), y^\prime) p^\prime(T_s(t), y^\prime|s, a) \\
		&\quad = \sum_{y^\prime \in Y'} \frac{\sum_{\tilde t: T_s(\tilde t) = T_s(t)}  \gamma_{s,\tilde t}^{y^\prime} q(\tilde t|s)}{q(T_s(t)|s)} \frac{\gamma_{st}^{y^\prime} \phi_{st}(y|y^\prime) q(t|s)}{\sum_{\tilde t: T_s(\tilde t) = T_s(t)} \gamma_{s\tilde t}^{y^\prime} q(\tilde t|s)} f^\prime(y^\prime|s, T_s(t), a)q(T_s(t)|s) \\
		&\quad = \sum_{y^\prime \in Y'} \gamma_{st}^{y^\prime} \phi_{st}(y|y^\prime) f^\prime(y^\prime|s, T_s(t), a) q(t|s) \\
		&\quad = f(y|s,t,a)q(t|s) \\
		&\quad = p(t, y|s, a),
	\end{align*}
	and for any other $(t, y)$, both $p(t, y|s, a)$ and $\sum_{(t^\prime , y^\prime) \in S \times Y' } \gamma_s^{t^\prime, y^\prime}\phi_s(t,y|t^\prime, y^\prime) p^\prime(t^\prime, y^\prime|s, a)$ are $0$ by definition for every $a \in A$. Therefore, $\Pi$ is a weighted garbling of $\Pi^\prime$.
\end{proof}

\subsection{Proof of \autoref{prop:sn2}}
\label{proof:sn2}

\begin{proof}

	\textbf{(Proof of the first claim in Item 1)} Fix $\delta \in (0,1)$. Take any $v \in \bigcap_{s \in S} E_{SS}^\delta(s) \subset \mathbb{R}$. For each $s \in S$, there exist $v^+ _s \equiv \max E_{SS}^\delta(s)$ (by the compactness of $E_{SS}^\delta(s)$), and an SSE $\sigma^+_s = (\sigma^+_{i,s})_{i \in I}$ such that $v^+ _s = (1-\delta ) u (\alpha_s^+,s ) + \delta \sum_{(t,y) \in S \times Y} w^+_s (t,y) p (t,y | s, \alpha_s^+) $, where $\alpha_s^+ \equiv \sigma_{i,s}^+ (s)$ (note that $\sigma_{i,s}^+ (s)  = \sigma_{j,s}^+ (s)$ for any $i,j\in I$) and $w_s^+ (t,y) \in \mathbb{R}$ is the continuation payoff when the next state is $t$ and the public signal is $y$. Define $x_s (t,y):= (\delta/(1- \delta)) (w_s^+ (t,y) - v_s^+) + v- v_s^+$ for each $(t,y) \in S \times Y$. Observe that:
	\begin{align*}
		&u (\alpha_s^+,s ) + \sum_{(t,y) \in S \times Y} x_s (t,y) p (t,y | s, \alpha_s^+) \\ 
		&\quad =u (\alpha_s^+,s ) + \sum_{(t,y) \in S \times Y} \frac{\delta}{ 1- \delta} (w_s^+ (t,y) - v_s^+)p (t,y | s, \alpha_s^+) + v- v_s^+ \\
		&\quad =v.
	\end{align*}
	Since $\sigma^+_s$ is an SSE, all the incentive constraints for $\mathcal{P}_{SS}(+1; \Pi)$ are satisfied. 
	Observe that, for any $T \subseteq S$, permutation $\xi$ on $T$ and $\psi: T \to Y$,
	\begin{align*}
		\sum_{s \in T } x_s (\xi (s),\psi (s)) &= \sum_{s\in T}\left(   \frac{\delta}{ 1- \delta} (w_s^+ (\xi(s),\psi(s)) - v_s^+) + v- v_s^+ \right)	\\
		& = \frac{ \delta}{ 1- \delta} \sum_{ s \in T} \left(w_s^+ (\xi(s),\psi(s)) - v_{ \xi(s)} ^+ \right) + \sum_{ s\in T} (v- v_s^+) \leq 0.
	\end{align*}
	Hence, $(v,x)$ is feasible for $\mathcal{P}_{SS}(+1; \Pi)$, and therefore $v \leq k_{SS}(+1; \Pi)$. Similarly, we can show that $v \geq -k_{SS}(-1 ; \Pi)$. 
	
	\vspace{5mm}

	\noindent \textbf{(Proof of the second claim in Item 1)}	 Let $\underline{z} \in \mathbb{R}$ and $\bar{z} \in \mathbb{R}$ be such that $-k_{SS}(-1; \Pi) < \underline{z}  < \bar{z}  < k_{SS}(+1; \Pi)$. Let $Z \equiv [\underline{z}, \bar{z}]$. Then, there exist $\epsilon_0 >0$, $v^{ \lambda} \in \mathbb{R}$ and $x^{ \lambda} = (x^\lambda_s)_s \in \mathbb{R}^{|S|^2 \times |Y|}$ for each $\lambda \in \{ -1, +1\}$ such that 1) $(v^{ \lambda} ,x^{ \lambda} )$ is feasible in $\mathcal{P}_{SS} (\lambda; \Pi)$ for each $\lambda \in \{-1, +1\}$, and 2) for any $z \in Z$, $v^{-1} + \epsilon_0 <z <v^{+1}-\epsilon_0$. Let 
	$$\kappa_0 \equiv \max \left \{  \max_{ \lambda \in \{ -1,+1\}, (s, t, y) \in S^2 \times Y } |x_s^\lambda(t, y)|, \max_{ \lambda \in \{ -1,+1\}, z \in Z } |z-v^\lambda|   \right\}.$$

	Take $n \in \mathbb{N}$ such that 
	$$\epsilon_0 (n-1)/2 - 2 \kappa_0 |S|>0.$$

	Let $\bar{\delta}<1$ be large enough so that for any $\delta \geq \bar{\delta}$, 1) $(n/2)^2 (1-\delta) \leq |S|$, and 2) $1-\delta^{n-1} \geq (n-1) (1-\delta)/2$. Given $w: H^n \to \mathbb{R}$, we denote by $\Gamma^n (s,w; \delta)$ the $\delta$-discounted $(n-1)$-stage game with final payoffs $w$ and initial state $s$.

	\begin{lem}
		\label{lem:nn1}
		For any $\lambda \in \{ -1, +1\}$, $z \in Z$, and $\delta \geq \bar{\delta}$, there exist continuation payoffs $w : H^n \to \mathbb{R}$ such that the following conditions hold:
		\begin{enumerate}
			\item For each $s \in S$, $z$ is an SSE payoff of the game $\Gamma^n(s,w ;\delta)$.
			\item For every $h \in H^n$, $\lambda w (h) < \lambda z$. 
		\end{enumerate}
	\end{lem}

	\begin{proof}[Proof of \autoref{lem:nn1}]
		
		Let $\lambda \in \{ -1, +1\}$. Following HSTV, define
		$$w(h^n) :=z +  \frac{ 1- \delta^{n-1}}{ \delta^{n-1}} (z- v^{ \lambda} ) + \frac{ 1- \delta}{ \delta^{n-1}} \sum_{k=1}^{n-1} \delta^{k-1} x^{ \lambda}_{s^k} (s^{k+1}, y^{k}), \quad \forall h^n \in H^n.$$

		Let $(\alpha_s)_s \in (\Delta (B))^{ |S|}$ be a Markov strategy such that $((\alpha_s)_s, v^{\lambda},x^{\lambda})$ is feasible in $\mathcal{P}_{SS} (\lambda; \Pi)$. Consider the strongly symmetric strategy profile in which each player uses $(\alpha_s)_s$. By the construction of $w$, we can show that there is no profitable one-shot deviation, so it is an SSE of $\Gamma^n(s,w ;\delta)$.

		To prove the second item, note that $\lambda \left(z- v^{ \lambda}\right) < - \epsilon_0$. HSTV shows (in their Lemma 5) that $((1-\delta)/\delta^{n-1})\sum_{k=1}^{n-1} \delta^{k-1} \lambda x^{ \lambda}_{s^{k}} (s^{k+1}, y^{k})$ is bounded above by $2 |S| \kappa_0 (1-\delta)/\delta^{n-1}$. Then, the sum of the second term and the third term of $\lambda w (h)$ is bounded above by $-((1- \delta)/ \delta^{n-1}) \left(((1- \delta^{n-1})/(1- \delta))\epsilon_0 - 2 |S| \kappa_0  \right)$, which is negative for our choice of $n$ and $\delta$. Therefore, we have $\lambda w (h) < \lambda z$ for every $h \in H^n$.
	\end{proof}

	For each $z \in Z$ and $\delta$, define $\hat{w}_z^{\delta} :H^n \to \mathbb{R} $ by
	\begin{equation}
		\label{eq:ooo1}
		\hat{w}_z^{\delta} :=\begin{cases}
			w^{+1, \delta}_z, & \text{if $z \in \left[ \frac{\bar{z} + \underline{z}}{2}, \bar{z} \right]$}\\
			w^{-1, \delta}_z , & \text{otherwise}
		\end{cases},
	\end{equation}
	where $w^{\lambda, \delta}_z$ is the continuation payoff in \autoref{lem:nn1} for $\lambda$, $z$ and $\delta$.

	Let $\bar{\bar{\delta}} \in (0,1)$ be large enough so that $((1- \delta^{n-1})/\delta^{n-1})2\kappa_0< (\bar{z}- \underline{z})/2$ for any $\delta \geq \bar{\bar{\delta}}$. Then, observe that for any $\delta \geq \max \{ \bar{\delta}, \bar{\bar{\delta}} \}$, $\hat{w}_z^\delta$ satisfies Item 1 and 2 of \autoref{lem:nn1}, and, in addition, $\hat{w}_z^\delta(h^n) \in Z$ for any $h^n \in H^n$, because 
	\begin{equation*}
		|w_z^{\lambda, \delta} (h^n) - z| \leq  \frac{1- \delta^{n-1}}{\delta^{n-1}}|z-v| + \frac{ 1- \delta}{ \delta^{n-1}} \sum_{k=1}^{n-1} \delta^{k-1} |x_{s^{k}}^\lambda (s^{k+1},y^k)|
		\leq \frac{1- \delta^{n-1}}{\delta^{n-1}}2\kappa_0   <\frac{ \bar{z}- \underline{z}}{2}
	\end{equation*}
	(for instance, if $z \in \left[ (\bar{z}+ \underline{z})/2, \bar{z} \right]$, then $\hat{w}_z^\delta(h^n) < z$ and $|\hat{w}_z^\delta(h^n)-z| < (\bar{z}-\underline{z})/2$ so $\hat{w}_z^\delta (h^n) \in Z$ for any $h^n \in H^n$).
	\vspace{2mm}

	Now consider the original stochastic game. Fix $\delta \geq  \max\{  \bar{\delta}, \bar{ \bar{\delta}}\}$. Take any $z \in Z$. We define, inductively, a strongly symmetric strategy profile $\sigma: H \to (\Delta (B))^N$ and continuation payoffs $w: \bigcup_{k=1}^\infty H^{(n-1)(k-1)+1} \to Z$ that achieve $z$ as follows:

	Let $w(h)  = z$ for any $h= (s^1) \in H^1$.

	For any $k \in \{ 1, 2, \dots \}$ and history $h \in H^{(n-1)(k-1)+1}$ (i.e., the history at the initial stage of the $k$-th ``block''), given that $w(h) \in Z$, let $\sigma$ prescribe the same action as the Markov strategy profile in the proof of \autoref{lem:nn1} for $\Gamma^n (s^{ (n-1)(k-1)+1}, \hat{w}_{w(h)}^\delta;\delta)$, where $\hat{w}_{w(h)}^\delta$ is defined as in \eqref{eq:ooo1}, during this block.

	For any history $\tilde{h} \in H^{(n-1)k+1}$ whose restriction to the first $(n-1)(k-1)+1$ periods is equal to $h$, let $w(\tilde{h}) = \hat{w}_{w(h)} (n \tilde{h})$, where $n\tilde{h}$ is the restriction of $\tilde{h}$ to the last $n$ periods. Then, by our choice of $\delta$, $w(\tilde{h}) \in Z$.

	By the one-shot deviation principle and Item 1 in \autoref{lem:nn1}, the strategy profile $\sigma$ is an SSE of the entire (infinite-horizon) game, with an average discounted payoff of $z$. Thus, $Z \subseteq E^\delta_{SS}(s)$ for any $s \in S$ and any $\delta \geq \max\{ \bar{\delta}, \bar{\bar{\delta}}\}$.

	We note that the same proof applies when we restrict our attention to the pure strategy SSE payoff sets. We can simply replace $\alpha_s \in \Delta(B), s \in S,$ with a pure strongly symmetric Markov strategy $a_s \in B, s \in S,$ in the above proof.
\end{proof}

\subsection{Proof of \autoref{strictWG}}
\label{proof:strictWG}

\begin{proof}
	Suppose that $\left[\underline{L}^\Pi_{s,k}, \overline{L}^\Pi_{s,k} \right] \subset \left(\underline{L}^{\Pi^\prime}_{s,k}, \overline{L}^{\Pi^\prime}_{s, k} \right)$ for every $s \in S$ and $k \in \{1,N\}$.

	Fix $s \in S$. First, we consider the full effort action profile. Since $\underline{L}^{\Pi^{\prime}}_{s,N} < 1 <\overline{L}^{\Pi^{\prime}}_{s,N}$ holds, $\Pi^\prime$ is informative at $N$. This implies that $(1, 1)$ lies in the interior of the positive cone spanned by $(p^\prime(\overline{t}^\prime_{s,N}, \overline{y}^\prime_{s,N}|s, N), p^\prime(\overline{t}^\prime_{s,N}, \overline{y}^\prime_{s,N}|s, N-1))$ and $(p^\prime(\underline{t}^\prime_{s,N}, \underline{y}^\prime_{s,N}|s, N), p^\prime(\underline{t}^\prime_{s,N}, \underline{y}^\prime_{s,N}|s, N-1))$ in $\mathbb{R}^2_+$. Hence, there exist unique $\overline{\gamma}_{s, N} > 0$ and $\underline{\gamma}_{s, N} > 0$ that satisfy the following equation for $k = N, N-1$. 
	\[
	\overline{\gamma}_{s, N} p^\prime(\overline{t}^\prime_{s,N}, \overline{y}^\prime_{s,N}|s, k) + \underline{\gamma}_{s, N} p^\prime(\underline{t}^\prime_{s,N}, \underline{y}^\prime_{s,N}|s, k)= 1.
	\]

	Consider the following auxiliary monitoring structure with two possible signals $(\overline{t}^\prime_{s,N}, \overline{y}^\prime_{s,N})$ and $(\underline{t}^\prime_{s,N}, \underline{y}^\prime_{s,N})$, where $(\overline{t}^\prime_{s,N}, \overline{y}^\prime_{s,N})$ is observed with probability $\overline{\gamma}_{s, N} p^\prime(\overline{t}^\prime_{s,N}, \overline{y}^\prime_{s,N}|s, k)$ and $(\underline{t}^\prime_{s,N}, \underline{y}^\prime_{s,N})$ is observed with probability $\underline{\gamma}_{s, N} p^\prime(\underline{t}^\prime_{s,N}, \underline{y}^\prime_{s,N}|s, k)$ for $k=N, N-1$. We show that $\Pi$ is a garbling of this auxiliary monitoring structure when we restrict attention to $k \in \{N-1, N\}$. Consider any full-support prior distribution on $N$ and $N-1$, say $(0.5, 0.5)$, and compute the posterior belief on $N$ given each signal. The posterior beliefs given $(\overline{t}^\prime_{s,N}, \overline{y}^\prime_{s,N})$ and $(\underline{t}^\prime_{s,N}, \underline{y}^\prime_{s,N})$ are $\overline{L}^{\Pi^{\prime}}_{s, N}/ (\overline{L}^{\Pi^{\prime}}_{s,N}  + 1)$ and $\underline{L}^{\Pi^{\prime}}_{s, N}/(\underline{L}^{\Pi^{\prime}}_{s, N} + 1)$, respectively. Similarly, the posterior belief given any $(t,y)$ for $\Pi$ is $L^\Pi_{s, N} (t,y)/(L^\Pi_{s, N} (t,y) + 1)$. This posterior belief is in $[\underline{L}^{\Pi}_{s, N}/(\underline{L}^{\Pi}_{s, N} + 1), \overline{L}^{\Pi}_{s, N}/(\overline{L}^{\Pi}_{s, N} + 1)]$, 
	hence in $(\underline{L}^{\Pi^{\prime}}_{s, N}/(\underline{L}^{\Pi^{\prime}}_{s, N} + 1), \overline{L}^{\Pi^{\prime}}_{s, N}/(\overline{L}^{\Pi^{\prime}}_{s, N} + 1))$ by assumption. Given that the expected posterior belief is $0.5$ for both distributions, it follows that the two-point distribution of posterior beliefs generated by this auxiliary monitoring structure is second-order stochastically dominated by the distribution of posterior beliefs for $\Pi$; in particular, the former is a mean-preserving spread of the latter. By the standard result for Blackwell experiments (interpreting $N$ and $N-1$ as binary states and the monitoring structures as experiments), we can find $\phi_{s,N}(\cdot, \cdot|\overline{t}^\prime_{s,N}, \overline{y}^\prime_{s,N}),  \phi_{s,N}(\cdot, \cdot|\underline{t}^\prime_{s,N}, \underline{y}^\prime_{s,N}) \in \Delta(S \times Y)$ that satisfy, for $k=N$ and $N-1$,
	\begin{multline*}
		p(t,y|s,k) = \phi_{s,N}(t,y|\overline{t}^\prime_{s,N}, \overline{y}^\prime_{s,N}) \overline{\gamma}_{s, N} p'(\overline{t}^\prime_{s,N}, \overline{y}^\prime_{s,N}|s, k)
		\\+ \phi_{s,N}(t,y|\underline{t}^\prime_{s,N}, \underline{y}^\prime_{s,N}) \underline{\gamma}_{s, N} p'(\underline{t}^\prime_{s,N}, \underline{y}^\prime_{s,N}|s, k), \quad \forall (t,y) \in S \times Y.
	\end{multline*}
	Note that both $\phi_{s,N}(t,y|\overline{t}^\prime_{s,N}, \overline{y}^\prime_{s,N})$ and $\phi_{s,N}(t,y|\underline{t}^\prime_{s,N}, \underline{y}^\prime_{s,N})$ must be strictly positive for every $(t,y)$ in the support of $p(\cdot, \cdot|s,k)$ for $k= N, N-1$, because $p(t,y|s,N)/p(t,y|s,N-1)$ is strictly between $\underline{L}^{\Pi^{\prime}}_{s,N} = p'(\underline{t}^\prime_{s,N}, \underline{y}^\prime_{s,N}|s, N)/p'(\underline{t}^\prime_{s,N}, \underline{y}^\prime_{s,N}|s, N-1)$ and $\overline{L}^{\Pi^{\prime}}_{s,N} = p'(\overline{t}^\prime_{s,N}, \overline{y}^\prime_{s,N}|s, N)/p'(\overline{t}^\prime_{s,N}, \overline{y}^\prime_{s,N}|s, N-1)$ by assumption. Denote this common support at state $s$ by $D_{s,N} \subseteq S \times Y$.

	Similarly, for the case of the no-effort profile, we can find $\overline{\gamma}_{s, 0} > 0, \underline{\gamma}_{s, 0}>0$, and $ \phi_{s,0}(\cdot,\cdot |\overline{t}^\prime_{s, 1}, \overline{y}^\prime_{s, 1}), \phi_{s,0}(\cdot,\cdot|\underline{t}^\prime_{s, 1}, \underline{y}^\prime_{s, 1})$ whose support is the same $D_{s,0} \subseteq S \times Y$, and the above garbling equations are satisfied for $k=0,1$. Note that the same construction works for every $s$ in both cases.

	Now we can show that $\Pi$ is a strict $P$-weighted garbling of $\Pi^\prime$ with respect to $\{0, N\}$. To see this for $k=N$, set $\gamma_{s,N}^{\overline{t}^\prime_{s,N}, \overline{y}^\prime_{s,N}} = \overline{\gamma}_{s, N}$, $\gamma_{s,N}^{\underline{t}^\prime_{s,N}, \underline{y}^\prime_{s,N}} = \underline{\gamma}_{s, N}$, and set $\gamma_{s,N}^{t',y'} = 0$ for any $(t', y') \notin \{(\overline{t}^\prime_{s,N}, \overline{y}^\prime_{s,N}), (\underline{t}^\prime_{s,N}, \underline{y}^\prime_{s,N})\}$. Also, use the above $\phi_{s,N}(t,y|\overline{t}^\prime_{s, N}, \overline{y}^\prime_{s, N})$ and $\phi_{s,N}(t,y|\underline{t}^\prime_{s, N}, \underline{y}^\prime_{s, N})$ for garbling given $(\overline{t}^\prime_{s, N}, \overline{y}^\prime_{s, N})$ and $(\underline{t}^\prime_{s,N}, \underline{y}^\prime_{s,N})$ respectively and assign an arbitrary $\phi_{s,N}(\cdot, \cdot|t', y') \in \Delta(S \times Y)$ with support $D_{s,N}$ for each $(t',y') \notin \{(\overline{t}^\prime_{s,N}, \overline{y}^\prime_{s,N}), (\underline{t}^\prime_{s,N}, \underline{y}^\prime_{s,N})\}$. It is straightforward to show that they satisfy the requirement of strict $P$-weighted garbling. A similar construction works for the case with $k=0$. This proves one direction.

	Conversely, suppose that $\Pi$ is a strict $P$-weighted garbling of $\Pi^\prime$ with respect to $\{0, N\}$ and $\Pi^\prime$ is informative at $\{0, N\}.$ Consider the full-effort action profile. For any $s \in S$, there exist nonnegative weight $\gamma_{s,N}^{t', y'} \geq 0$ and  $\phi_{s,N}(\cdot,\cdot|t', y') \in \Delta (S \times Y)$ for every $(t', y')$ with some common support $D_{s,N} \subseteq S \times Y$ such that the following equations are satisfied:
	\begin{align*}
		p(t,y|s,k) = \sum_{(t',y') \in S \times Y^\prime} \gamma_{s,N}^{t',y'} \phi_{s,N}(t, y| t', y')p'(t',y'|s,k), \ \forall (t,y) \in S \times Y, \ \forall k \in \{N, N-1\}.
	\end{align*}
	It then follows that, for every $(t, y)\in D_{s,N}$,
	\[
	\frac{p(t,y|s,N)}{p(t,y|s,N-1)} = \frac{\sum_{(t',y') \in S \times Y^\prime} \gamma_{s,N}^{t',y'} \phi_{s,N}(t, y| t', y')p'(t',y'|s,N)}{\sum_{(t',y') \in S \times Y^\prime} \gamma_{s,N}^{t',y'} \phi_{s,N}(t, y| t', y')p'(t',y'|s,N-1)}.
	\]
	
	If $\underline{L}^{\Pi}_{s,N} = \overline{L}^{\Pi}_{s,N} = 1$, then the result directly follows from the informativeness of $\Pi^\prime$. So suppose not and $-\infty < \underline{L}^{\Pi}_{s,N} < \overline{L}^{\Pi}_{s,N} < \infty$ holds. Then, the above equations imply that there must exist some $(t_1',y_1'), (t_2',y_2') \in S \times Y'$ such that $\gamma_{s,N}^{t_1',y_1'}, \gamma_{s,N}^{t_2',y_2'} > 0$ and $p'(t_1',y_1'|s,N)/p'(t_1',y_1'|s,N-1) < p(t,y|s,N)/p(t,y|s,N-1) < p'(t_2',y_2'|s,N)/p'(t_2',y_2'|s,N-1)$ for every $(t,y)$ in the support $D_{s,N}$ of $p(\cdot, \cdot|s,k), k=N-1, N$. Since $\underline{L}^{\Pi^\prime}_{s,N} \leq p'(t_1',y_1'|s,N)/p'(t_1',y_1'|s,N-1)$ and $\overline{L}^{\Pi^\prime}_{s,N} \geq p'(t_2',y_2'|s,N)/p'(t_2',y_2'|s,N-1)$, we obtain $\left[\underline{L}^{\Pi}_{s,N}, \overline{L}^{\Pi}_{s,N}\right]\subset \left(\underline{L}^{\Pi^\prime}_{s,N}, \overline{L}^{\Pi^\prime}_{s,N}\right)$. 
	
	The proof for the no-effort case is almost identical and is therefore omitted. This proves the other direction.
\end{proof}

\subsection{Proof of \autoref{thm:invariant}}
\label{proof:them4}
\begin{proof}
	
	For each $\lambda \in \Lambda$, let $I^\lambda_{>} \subseteq I$ be the set of players with a positive weight (i.e., $\lambda_i > 0)$, $I^\lambda_{<} \subseteq I$ be the set of players with a negative weight, and $I^\lambda_{0} \subseteq I$ be the set of players with zero weight. Denote the set of players with nonzero weight by $I^\lambda_{\neq 0} = I^\lambda_{>} \cup I^\lambda_{<}$.

	An action profile $\alpha = (\alpha_s)_{s \in S}$, $\alpha_s \in \prod_{i \in I} \Delta (A_i)$ for each $s \in S$, is called \textit{admissible} if, for each $s \in S$ and $i \in I$, $u_i((\alpha^{\prime}_{i,s}, \alpha_{-i, s}), s) \leq u_i(\alpha_s, s)$ holds for any $\alpha^{\prime}_{i,s} \in \Delta (A_i)$ such that $p(t,y|s, \alpha_s) = p(t,y|s, (\alpha^{\prime}_{i,s}, \alpha_{-i, s}))$ for all $(t,y) \in S \times Y$.\footnote{This is the standard condition that requires any undetectable deviation to be unprofitable.}

	Given each $\lambda \in \Lambda$, the dual problem of the score maximization problem in \Cref{subsec:4.1} is as follows (see Section 3 of \cite{HTV_2014_GEB}). 
	\[
	\sup_{\alpha \ \text{admissible}} \inf_{\substack{\hat \alpha \in D_\lambda(\alpha) \\ \beta \in B(\hat \alpha)}} \sum_{s \in S, i \in I} \lambda_i \beta_s u_i\left(\left(\hat \alpha_{i,s},\alpha_{-i,s}\right) , s\right),
	\]
	where
	\begin{itemize}
		\item $D_\lambda(\alpha)$ is the set of $\hat \alpha = (\hat \alpha_s)_{s \in S} \in \left(\prod_{i \in I}\mathbb{R}^{\left|A_i\right|}\right)^{\left|S\right|}$ such that for each $s \in S$ and each $i \in I^\lambda_{\neq 0}$:\footnote{$D_{\lambda}(\alpha)$ is not empty, as $\alpha$ itself is a member of $D_\lambda(\alpha)$.}
		\begin{itemize} 
			\item $\sum_{a_i \in A_i} \hat \alpha_{i,s}(a_i) = 1$ and $\hat \alpha_{i,s}(a_i) \left\{  \begin{array}{ll}
				\leq 0 \ \text{if} \ i \in I^\lambda_> \ \text{and} \ \alpha_{i,s}(a_i) = 0,\\
				\geq 0 \ \text{if} \ i \in I^\lambda_< \ \text{and} \ \alpha_{i,s}(a_i) = 0,\\
				\text{unrestricted} \ \text{if} \ \alpha_{i,s}(a_i) > 0
			\end{array}
			\right.$

			\item $\hat p\left(t,y|s\right) = p\left(t,y|s, \left(\hat \alpha_{i,s}, \alpha_{-i,s}\right)\right), \forall (t,y)$ for every $i \in I^\lambda_{\neq 0}$ for some $\hat p\left(\cdot, \cdot|s\right) \in \Delta(S \times Y)$.
		\end{itemize}
		\item $B(\hat \alpha)$ is the set of invariant distributions of the Markov process over $S$ generated by $\hat p\left(t,y|s\right)$.
		
	\end{itemize}
	Note that $\hat \alpha_{i,s}$ is an ``extended'' mixed action, i.e., $\hat \alpha_{i,s}(a_i)$ can take a negative value. The definition of $u_i$ and $p$ is extended naturally to allow for such extended mixed actions.

	We prove that the value of this problem is greater for $(q',\Pi')$. We first show that the set of admissible $\alpha = (\alpha_s)_s$ is larger with $(q', \Pi')$. Take any $\alpha$ that is admissible for $(q,\Pi)$ and any $\alpha'_{i,s}$ that satisfies $p' (t,y'|s,\alpha_s) = p' (t,y' |s, (\alpha'_{i,s}, \alpha_{-i,s}))$ for all $(t,y') \in S \times Y'$. Then, 
	\begin{align*}
		p(t,y|s, \alpha_s) &= \sum_{(t',y') \in S \times Y'} \gamma_s^{t',y'}  \phi_s (t,y|t',y') p' (t',y' |s,\alpha_s) \\
		&= \sum_{(t',y') \in S \times Y'} \gamma_s^{t',y'}  \phi_s (t,y|t',y')  p' (t',y' |s, (\alpha'_{i,s}, \alpha_{-i,s}))\\
		&= p (t,y|s, (\alpha'_{i,s}, \alpha_{-i,s})).
	\end{align*}
	Since $\alpha$ is admissible for $(q,\Pi)$, $u_i ( (\alpha'_{i,s},\alpha_{-i,s}),s) \leq u_i (\alpha_s,s)$. Therefore, $\alpha$ is admissible for $(q',\Pi')$.

	Next, we show that the value of the minimization problem is smaller for $(q,\Pi)$ given any admissible $\alpha$. We first observe the following.

	\begin{lem}
		\label{lem:o2}
		Suppose that for any pure Markov strategy profile $q$ and $q'$ induce irreducible Markov chains with the same invariant distribution. Let $\tilde{\alpha}: S \to \mathbb{R}^{ |A|}$ be such that $\sum_{a \in A} \tilde{\alpha}_s (a) =1$ for each $s \in S$ (i.e., ``extended'' Markov strategy profile), and let $\tilde{\beta} = (\tilde{\beta}_s)_s \in \mathbb{R}^{|S|}$ be ``invariant'' under $\tilde{\alpha}$ with $q$ in the sense that $\sum_{ s \in S} \sum_{ a \in A} \tilde{\beta}_s   \tilde{\alpha}_s  (a) q (t|s,a)= \tilde{\beta}_t$ for every $t \in S$. Then, $\tilde{\beta}$ is invariant under $\tilde{\alpha}$ with $q'$.
	\end{lem}
	\begin{proof}
		See \Cref{proof:lemo2}.
	\end{proof}
	Thus, if $\hat{\beta}$ is invariant under $\tilde{\alpha}_s (a_i, a_{-i}) = \hat{\alpha}_{i,s} (a_i) \alpha_{-i,s} (a_{-i})$, $\forall s, a_i, a_{-i}$, with $q$, then it is so with $q'$. In particular, if $\hat{\alpha} \in D'_\lambda (\alpha) \cap D_\lambda (\alpha)$, then $B(\hat \alpha)=B'(\hat{\alpha})$, where $D_\lambda' (\alpha)$ and $B' (\hat{\alpha})$ are for $(q',\Pi')$.

	Also, note that the value of the objective function does not depend on the monitoring structure. Thus, it suffices to show that $D_\lambda^\prime(\alpha)$ for $(q',\Pi')$ is a subset of $D_\lambda(\alpha)$ for $(q,\Pi)$. 
	
	Take any $\hat{\alpha} \in D'_\lambda(\alpha)$. We can show in a similar way to the above proof for admissibility that, at every $s \in S$,  $p (\cdot, \cdot|s, (\hat{\alpha}_{i,s}, \alpha_{-i,s}))=p (\cdot, \cdot|s, (\hat{\alpha}_{j,s}, \alpha_{-j,s})) = \hat{p} (\cdot, \cdot|s) \in \Delta(S \times Y)$ holds for all $ i,j \in I^\lambda_{\neq 0}$, where $\hat{p}(t,y|s ) \equiv \sum_{(t',y') \in S \times Y'} \gamma_s^{t',y'}  \phi_s (t,y|t',y') \hat{p}^\prime(t',y'|s ) $  for each $(t, y) \in S \times Y$. This proves $\hat{\alpha} \in D_\lambda(\alpha)$.
\end{proof}

\subsection{Proof of \autoref{lem:o2}}
\label{proof:lemo2}
\begin{proof}

	For each pure Markov strategy profile $\mathbf{a} =(a_{s})_s \in A^{|S|}$, let $\pi ( \mathbf{a})  = (\pi_{s} (\mathbf{a}))_{s\in S}$ denote the unique invariant distribution induced by $ \mathbf{a}$ and $q$. The following observation is the key to the proof.

	\begin{lem}
		\label{lem:o3}
		Let $\tilde{\alpha}: S \to \mathbb{R}^{ |A|}$ be such that $\sum_{ a \in A} \tilde{\alpha}_s (a)=1$ for each $s \in S$ and let $\tilde{\beta}= (\tilde{\beta}_s)_s \in \mathbb{R}^{ |S|}$ be invariant under $\tilde{\alpha}$ with $q$. Then, there exists $k: A^{ |S|} \to \mathbb{R}$ such that
		\begin{equation}
			\label{eq:nnnnnn1}
			\tilde{\beta}_{s} \tilde{\alpha}_{s}  (a) = \sum_{\mathbf{a} \in A^{ |S|} : a_{s} = a} k (\mathbf{a}) \pi_{s} ( \mathbf{a}), \quad \forall s \in S, a \in A.
		\end{equation}
	\end{lem}

	We prove \autoref{lem:o3} after the proof of \autoref{lem:o2}.

	Given this lemma, the remainder of the proof is relatively straightforward. Observe that for any $t \in S$,
	\begin{multline*}
		\sum_{s \in S} \sum_{a \in A}  \tilde{\beta}_s  \tilde{\alpha}_{s} (a) q' ( t|s, a)= \sum_{s \in S} \sum_{\mathbf{a} \in A^{ |S|} }  k (\mathbf{a}) \pi_{s} ( \mathbf{a})q' (t|s, a_{s})\\
		=  \sum_{\mathbf{a} \in A^{ |S|}  }  k (\mathbf{a}) \sum_{s \in S} \pi_{s} ( \mathbf{a})q' (t|s, a_{s})=  \sum_{\mathbf{a} \in A^{ |S|}   }  k (\mathbf{a})  \pi_{t} ( \mathbf{a}) =\tilde{\beta}_t,
	\end{multline*}
	where the third equality follows from the assumption of the identical invariant distribution for $q$ and $q'$ and the fourth equality follows from \eqref{eq:nnnnnn1}.
	Therefore, $\tilde{\beta}$ is invariant under $\tilde{\alpha}$ with $q'$.
\end{proof}

\begin{proof}[Proof of \autoref{lem:o3}]
	In the following, we use the following fact: Let $\mathbf{A}$ be $(m \times n)$-matrix and $b \in \mathbb{R}^{m}$. Then $\mathbf{A}x = b$ for some $x \in \mathbb{R}^n$ if and only if for any row vector $z \in \mathbb{R}^m$ if $z \mathbf{A} =\mathbf{0} \in \mathbb{R}^n$, $zb=0$.

	Let $\mathbf{A}$ be a $ ( |S| \times |A|) \times (|A|^{ |S|})$ matrix whose $(s,a)$-row has $\pi_s ( \mathbf{a})$ in the column corresponding to $\mathbf{a} \in A^{|S|}$ if $a_{s} = a$, and $0$ otherwise. Also, let $x = (k (\mathbf{a}))_{ \mathbf{a} \in A^{|S|}}$, and $b = (  \tilde{\beta}_s    \tilde{\alpha}_{s} (a))_{s,a}$. Then $\mathbf{A}x =b$ represents \eqref{eq:nnnnnn1}. 
	
	Let $z = (z_s (a))_{ s \in S, a \in A} \in \mathbb{R}^{ |S| \times |A|}$. Then, $z \mathbf{A}=0$ corresponds to the condition that for each pure Markov strategy profile $\mathbf{a} = (a_{s})_s$, 
	\begin{equation}
		\label{eq:o2}
		\sum_{s \in S} z_{s} (a_{s}) \pi_{s} (\mathbf{a})= 0.
	\end{equation}
	
	For each $t \in S$ and $a\in A$, define 
	$$\tilde{z}_s^t (a):=\begin{cases}
		-\sum_{\tilde{s} \neq t} q (\tilde{s}|s, a), & \text{if $s=t$}\\
		q (t|s, a), & \text{otherwise}
	\end{cases}
	$$
	Thus, for pure Markov strategy profile $\mathbf{a} = (a_s)_s$,  $\tilde{z}_s^t (a_{s}) \pi_s (\mathbf{a})$ is the total ``outflow'' probability from state $t$ if $s=t$; otherwise, it is the ``inflow'' probability to state $t$ from $s \neq t$, given that $\mathbf{a}$ is played. So, 
	\begin{equation}
		\label{eq:o3}
		\sum_{ s \in S} \tilde{z}_s^t (a_{s}) \pi_s (\mathbf{a}) = 0 	
	\end{equation}
	by the definition of invariant distribution. Thus, for each $t$, $\tilde{z}^t = (\tilde{z}_s^t(a))_{s, a}$ is a solution of \eqref{eq:o2}. 
	
	Choose an arbitrary state $\hat{s} \in S$. We claim that for any solution $z = (z_s (a))_{s, a}$ that solves \eqref{eq:o2}, there exist $\lambda^t \in \mathbb{R}$ for each $t \in S \setminus \{ \hat{s} \}$ such that 
	$$z_s (a) =\sum_{t \neq \hat{s} } \lambda^t \tilde{z}_s^t (a), \quad \forall s \in S, a \in A.$$ 
	
	To see this, we first observe that for each pure Markov strategy profile $\mathbf{a} = (a_{s})_s$, the set $\{ (\tilde{z}_s^t (a_{s}))_s : t \neq \hat{s} \} $ is linearly independent. Recall that $ Q (\mathbf{a})^\top \pi (\mathbf{a})= \pi (\mathbf{a})$ or equivalently $ (Q (\mathbf{a})^\top -\mathbb{I}) \pi (\mathbf{a})=0$, where $Q (\mathbf{a})$ is the transition matrix (with $q(\cdot |s,a_s)$ as its $s$th row) given $\mathbf{a}$ and $\mathbb{I}$ is the identity matrix. Then by the Perron-Frobenius theorem for irreducible matrices, the null space of $(Q (\mathbf{a})^\top -\mathbb{I})$ is one-dimensional. Thus, by the rank-nullity theorem, $Q (\mathbf{a})^\top -\mathbb{I}$ has rank $|S|-1$. Note that $Q (\mathbf{a})^\top -\mathbb{I}$ contains $\tilde{z}^t (\mathbf{a}) \equiv (\tilde{z}^t_s (a_{s}))_s$ in the row corresponding to state $t$. In addition, for each $s$, $\sum_{ t\in S} \tilde{z}^{t}_s (a_{s}) = -\sum_{ \tilde{s} \neq s} q(\tilde{s} |s, a_{s}) + \sum_{t \neq s} q(t|s, a_{s}) =  0 $, so $\tilde{z}^{ \hat{s}} (\mathbf{a})   =- \sum_{ t \neq \hat{s}}\tilde{z}^t (\mathbf{a}) $. Thus, $\{ \tilde{z}^t (\mathbf{a}): t \neq \hat{s} \}$ should be linearly independent in order for $(Q(\mathbf{a})^\top -\mathbb{I})$ to have rank $|S|-1$.

	Consider any solution $z = (z_s (a))_{s, a}$ of $\eqref{eq:o2}$. Pick an arbitrary pure Markov strategy profile $\mathbf{a} = (a_{s})_s$. Note that the set of variables $(z_s (a_{s}))_s \in \mathbb{R}^{ |S|}$ that satisfy \eqref{eq:o2} for $\mathbf{a}$ (not all equations) has at most $|S|-1$ dimension. From \eqref{eq:o3} for $\mathbf{a}$, we know that for each $t \in S \setminus \{ \hat{s}\}$, $(\tilde{z}^t_s (a_{s}))_s$ is a solution for \eqref{eq:o2} for $\mathbf{a}$, and, from the previous paragraph, we know that they are linearly independent. Therefore, $(z_s (a_{s}))_s$ must be a linear combination of $(\tilde{z}^t_s (a_{s}))_s$, $t \in S \setminus \{ \hat{s}\}$, i.e., 
	$$z_s (a_{s}) = \sum_{ t \neq \hat{s}} \lambda^t \tilde{z}_s^t (a_{s}), \quad \forall s \in S$$
	for some $\lambda^t \in \mathbb{R}$ for each $t \in S \setminus \{ \hat{s} \}$. Now we show that this $\lambda^t$ satisfies $z_s (a) = \sum_{ t \neq \hat{s}} \lambda^t \tilde{z}_s^t (a)$ for any $a \in A$ and $s \in S$. Pick an arbitrary $s \in S$ and $a'_{s} \neq a_{s}$. Then, from equation \eqref{eq:o2} for pure Markov strategy profile $(a_{s}', a_{-s})$ (i.e., the pure Markov strategy profile that replaces $a_{s}$ with $a_{s}'$ only for state $s$ from $\mathbf{a}$), it follows that
	\begin{align*}
		&z_{s} (a_{s}')\pi_{s} ( (a_{s}', a_{-s})) + \sum_{s' \neq s} z_s (a_{s'})\pi_{s'} ( (a_{s}', a_{-s})) \\
		&\quad  =z_{s} (a_{s}')\pi_{s} ( (a_{s}', a_{-s}))+\sum_{s' \neq s} \sum_{t \neq \hat{s}} \lambda^t \tilde{z}_{s'}^t (a_{s'})\pi_{s'} ( (a_{s}', a_{-s}))\\
		&\quad = z_{s} (a_{s}')\pi_{s} ( (a_{s}', a_{-s}))+\sum_{t \neq \hat{s}} \lambda^t \sum_{s' \neq s} \tilde{z}_{s'}^t (a_{s'})\pi_{s'} ( (a_{s}', a_{-s}))\\
		&\quad =z_{s} (a_{s}')\pi_{s} ( (a_{s}', a_{-s}))+\sum_{t \neq \hat{s}} \lambda^t (- \tilde{z}_s^t (a_{s}') \pi_{s} ( (a_{s}', a_{-s}))) \\
		&\quad =0,
	\end{align*}
	where the third equality follows from \eqref{eq:o3} for pure Markov strategy profile $(a_{s}', a_{-s})$.
	By the irreducibility assumption, $\pi_s ( (a_{s}', a_{-s}))> 0$, so the above equation implies
	$$z_{s} (a_{s}') = \sum_{t \neq \hat{s}} \lambda^t  \tilde{z}_s^t (a_{s}').$$
	Since we picked an arbitrary $s$ and $a'_{s}$, we proved the claim.
	
	Next we claim that for any solution $z = (z_s (a))_{s,a}$ of \eqref{eq:o2}, $zb=0$. From our previous claim, let $z_s (a) = \sum_{t \neq \hat{s} } \lambda^t \tilde{z}_s^t (a)$ for each $s$ and $a$. Observe that
	$$zb=\sum_{t \neq \hat{s}} \lambda^t \sum_{s \in S}\sum_{a \in A} \tilde{z}_s^t (a) \tilde{\alpha}_{s} (a) \tilde{\beta}_s = 0,$$
	because for each $t \neq \hat{s}$,
	\begin{multline*}
		\sum_{a \in A} \tilde{z}_t^t (a) \tilde{\alpha}_{t} (a)\tilde{\beta}_t +\sum_{s \neq t} \sum_{a \in A} \tilde{z}_s^t (a) \tilde{\alpha}_{s} (a)\tilde{\beta}_s \\
		= \left(- \sum_{ \tilde{s} \neq t}\sum_{a \in A} q ( \tilde{s} |t,a)\tilde{\alpha}_{t} (a)  \right) \tilde{\beta}_t + \sum_{s \neq t}\left(   \sum_{a \in A} q (t|s, a) \tilde{\alpha}_{s} (a) \right)\tilde{\beta}_s =0,
	\end{multline*}
	where the last equality follows as $\tilde{\beta} = (\tilde{\beta}_s)_s$ is invariant under $\tilde{\alpha}$ with $q$.  
\end{proof}

\bibliography{StochasticGarblingMain}    

@article{RZ_2020_MOR,
	author = {Renault, J{\'e}r{\^o}me and Ziliott, Bruno},
	journal = {Mathematics of Operations Research},
	number = {3},
	pages = {889--895},
	publisher = {INFORMS},
	title = {Limit equilibrium payoffs in stochastic games},
	volume = {45},
	year = {2020}}

@article{HRSV_2010_OR,
	author = {H{\"o}rner, Johannes and Rosenberg, Dinah and Solan, Eilon and Vieille, Nicolas},
	journal = {Operations research},
	number = {4-part-2},
	pages = {1107--1115},
	publisher = {INFORMS},
	title = {On a Markov game with one-sided information},
	volume = {58},
	year = {2010}}

@article{PT_2017_ECMA,
	author = {P{\k{e}}ski, Marcin and Toikka, Juuso},
	journal = {Econometrica},
	number = {6},
	pages = {1921--1948},
	publisher = {Wiley Online Library},
	title = {Value of persistent information},
	volume = {85},
	year = {2017}}

@article{RMM_1986_ECMA,
	author = {Radner, Roy and Myerson, Roger and Maskin, Eric},
	journal = {The Review of Economic Studies},
	number = {1},
	pages = {59--69},
	publisher = {Wiley-Blackwell},
	title = {An example of a repeated partnership game with discounting and with uniformly inefficient equilibria},
	volume = {53},
	year = {1986}}

@article{Yamamoto_2009_JET,
	author = {Yamamoto, Yuichi},
	journal = {Journal of Economic Theory},
	number = {2},
	pages = {802--824},
	publisher = {Elsevier},
	title = {A limit characterization of belief-free equilibrium payoffs in repeated games},
	volume = {144},
	year = {2009}}

@article{EHO_2005_ECMA,
	author = {Ely, Jeffrey C and H{\"o}rner, Johannes and Olszewski, Wojciech},
	journal = {Econometrica},
	number = {2},
	pages = {377--415},
	publisher = {Wiley Online Library},
	title = {Belief-free equilibria in repeated games},
	volume = {73},
	year = {2005}}

@article{Mailath_Matthews_Sekiguchi_2002_BE,
	author = {Mailath, George J and Matthews, Steven A and Sekiguchi, Tadashi},
	journal = {The BE Journal of Theoretical Economics},
	number = {1},
	publisher = {De Gruyter},
	title = {Private strategies in finitely repeated games with imperfect public monitoring},
	volume = {2},
	year = {2002}}

@article{FL_2009_QJE,
	author = {Fudenberg, Drew and Levine, David K},
	journal = {The Quarterly Journal of Economics},
	number = {1},
	pages = {233--265},
	publisher = {MIT Press},
	title = {Repeated games with frequent signals},
	volume = {124},
	year = {2009}}

@article{AMP_1991_ECMA,
	author = {Abreu, Dilip and Milgrom, Paul and Pearce, David},
	journal = {Econometrica: Journal of the Econometric Society},
	pages = {1713--1733},
	publisher = {JSTOR},
	title = {Information and timing in repeated partnerships},
	year = {1991}}

@article{Sannikov_Skrzypacz_2010_ECMA,
	author = {Sannikov, Yuliy and Skrzypacz, Andrzej},
	journal = {Econometrica},
	number = {3},
	pages = {847--882},
	publisher = {Wiley Online Library},
	title = {The role of information in repeated games with frequent actions},
	volume = {78},
	year = {2010}}

@article{FL_2007_RED,
	author = {Fudenberg, Drew and Levine, David K},
	journal = {Review of Economic Dynamics},
	number = {2},
	pages = {173--192},
	publisher = {Elsevier},
	title = {Continuous time limits of repeated games with imperfect public monitoring},
	volume = {10},
	year = {2007}}

@techreport{Sugaya_Wolitzky_2022_WP,
	author = {Sugaya, Takuo and Wolitzky, Alexander},
	institution = {Working Paper},
	title = {Informational Requirements for Cooperation},
	year = {2022}}

@article{Sannikov_Skrzypacz_2007_AER,
	author = {Sannikov, Yuliy and Skrzypacz, Andrzej},
	journal = {American Economic Review},
	number = {5},
	pages = {1794--1823},
	publisher = {American Economic Association},
	title = {Impossibility of collusion under imperfect monitoring with flexible production},
	volume = {97},
	year = {2007}}

@article{Blackwell_1953_AMS,
	author = {Blackwell, David},
	journal = {The annals of mathematical statistics},
	pages = {265--272},
	publisher = {JSTOR},
	title = {Equivalent comparisons of experiments},
	year = {1953}}

@article{HTV_2014_GEB,
	author = {H{\"o}rner, Johannes and Takahashi, Satoru and Vieille, Nicolas},
	journal = {Games and Economic Behavior},
	pages = {70--83},
	publisher = {Elsevier},
	title = {On the limit perfect public equilibrium payoff set in repeated and stochastic games},
	volume = {85},
	year = {2014}}

@article{Sugaya_Wolitzky_2018_JPE,
	author = {Sugaya, Takuo and Wolitzky, Alexander},
	journal = {Journal of Political Economy},
	number = {6},
	pages = {2569--2607},
	publisher = {University of Chicago Press Chicago, IL},
	title = {Maintaining privacy in cartels},
	volume = {126},
	year = {2018}}

@article{FL_1994_JET,
	author = {Fudenberg, Drew and Levine, David K},
	journal = {Journal of Economic Theory},
	number = {1},
	pages = {103--135},
	publisher = {Elsevier},
	title = {Efficiency and observability with long-run and short-run players},
	volume = {62},
	year = {1994}}

@article{Kim_2019_IJGT,
	author = {Kim, Daehyun},
	journal = {International Journal of Game Theory},
	number = {1},
	pages = {267--285},
	publisher = {Springer},
	title = {Comparison of information structures in stochastic games with imperfect public monitoring},
	volume = {48},
	year = {2019}}

@article{Kloosterman_2015_JET,
	author = {Kloosterman, Andrew},
	journal = {Journal of Economic Theory},
	pages = {28--48},
	publisher = {Elsevier},
	title = {Public information in markov games},
	volume = {157},
	year = {2015}}

@article{Kandori_Obara_2006_ECMA,
	author = {Kandori, Michihiro and Obara, Ichiro},
	journal = {Econometrica},
	number = {2},
	pages = {499--519},
	publisher = {JSTOR},
	title = {Efficiency in repeated games revisited: The role of private strategies},
	volume = {74},
	year = {2006}}

@article{FLM_1994_ECMA,
	author = {Fudenberg, Drew and Levine, David and Maskin, Eric},
	journal = {Econometrica},
	number = {5},
	pages = {997--1039},
	publisher = {JSTOR},
	title = {The folk theorem with imperfect public information},
	volume = {62},
	year = {1994}}

@article{APS_1990_ECMA,
	author = {Abreu, Dilip and Pearce, David and Stacchetti, Ennio},
	journal = {Econometrica},
	number = {5},
	pages = {1041--1063},
	publisher = {JSTOR},
	title = {Toward a theory of discounted repeated games with imperfect monitoring},
	volume = {58},
	year = {1990}}

@inproceedings{Blackwell_1951,
	author = {Blackwell, David},
	booktitle = {Proceedings of the second Berkeley symposium on mathematical statistics and probability},
	pages = {93--102},
	title = {Comparison of experiments},
	volume = {1},
	year = {1951}}

@article{FY_2011_JET,
	author = {Fudenberg, Drew and Yamamoto, Yuichi},
	journal = {Journal of Economic Theory},
	number = {4},
	pages = {1664--1683},
	publisher = {Elsevier},
	title = {The folk theorem for irreducible stochastic games with imperfect public monitoring},
	volume = {146},
	year = {2011}}

@article{Dutta_1995_JET,
	author = {Dutta, Prajit K},
	journal = {Journal of Economic Theory},
	number = {1},
	pages = {1--32},
	publisher = {Elsevier},
	title = {A folk theorem for stochastic games},
	volume = {66},
	year = {1995}}

@article{HSTV_2011_ECMA,
	author = {H{\"o}rner, Johannes and Sugaya, Takuo and Takahashi, Satoru and Vieille, Nicolas},
	journal = {Econometrica},
	number = {4},
	pages = {1277--1318},
	publisher = {Wiley Online Library},
	title = {Recursive methods in discounted stochastic games: An algorithm for $\delta \to 1$ and a folk theorem},
	volume = {79},
	year = {2011}}

@article{Kandori_1992_RES,
	author = {Kandori, Michihiro},
	journal = {The Review of Economic Studies},
	number = {3},
	pages = {581--593},
	publisher = {Oxford University Press},
	title = {The use of information in repeated games with imperfect monitoring},
	volume = {59},
	year = {1992}}
\bibliographystyle{te}

\end{document}